\documentclass{article}

\usepackage{arxiv}

\usepackage[utf8]{inputenc} % allow utf-8 input
\usepackage[T1]{fontenc}    % use 8-bit T1 fonts
\usepackage{hyperref}       % hyperlinks
\usepackage{url}            % simple URL typesetting
\usepackage{booktabs}       % professional-quality tables
\usepackage{amsfonts}       % blackboard math symbols
\usepackage{nicefrac}       % compact symbols for 1/2, etc.
\usepackage{microtype}      % microtypography
\usepackage{lipsum}
\usepackage[pdftex]{graphicx}
\usepackage{epstopdf}
\usepackage{ifpdf}
\usepackage{cite}
\usepackage{amsmath}
\usepackage{array}
\usepackage[caption=false,font=normalsize,labelfont=sf,textfont=sf]{subfig}
\usepackage{fixltx2e}
\usepackage[noend]{algpseudocode}
\usepackage{algorithmicx,algorithm}
\usepackage{amssymb}
\usepackage{ntheorem}
\usepackage{boondox-calo}
\usepackage{upgreek}
\newcommand{\tabincell}[2]{\begin{tabular}{@{}#1@{}}#2\end{tabular}}
\newtheorem{theorem}{Theorem}[section]

\newtheorem{assumption}{Assumption}[section]

\newtheorem{proposition}{Proposition}[section]

\newtheorem{remark}{Remark}[section]

\newtheorem*{proof}{Proof}[section]
\newtheorem{lemma}{Lemma}[section]

\renewcommand{\thealgorithm}{\arabic{section}.\arabic{algorithm}}
\newcommand{\abs}[1]{\left\vert#1\right\vert}

\newcommand{\R}{\mathbb R}

\newcommand{\N}{\mathbb N}

\newcommand{\eqn}[1]{(\ref{#1})}

\newcommand{\sign}[1]{\mathrm{sgn}(#1)}
\newcommand{\ba}[1]{\begin{array}{#1}}
\newcommand{\ea}{\end{array}}
\newcommand{\beq}{\begin{equation}}
\newcommand{\eeq}{\end{equation}}
\newcommand{\beqar}{\begin{eqnarray}}
\newcommand{\eeqar}{\end{eqnarray}}
\newcommand{\bpm}{\begin{pmatrix}}
\newcommand{\epm}{\end{pmatrix}}
\def\comp{\ensuremath\mathop{\scalebox{.6}{$\circ$}}}
\makeatletter
\newenvironment{breakablealgorithm}
  {% \begin{breakablealgorithm}
   \begin{center}
     \refstepcounter{algorithm}% New algorithm
     \hrule height.8pt depth0pt \kern2pt% \@fs@pre for \@fs@ruled
     \renewcommand{\caption}[2][\relax]{% Make a new \caption
       {\raggedright\textbf{\ALG@name~\thealgorithm} ##2\par}%
       \ifx\relax##1\relax % #1 is \relax
         \addcontentsline{loa}{algorithm}{\protect\numberline{\thealgorithm}##2}%
       \else % #1 is not \relax
         \addcontentsline{loa}{algorithm}{\protect\numberline{\thealgorithm}##1}%
       \fi
       \kern2pt\hrule\kern2pt
     }
  }{% \end{breakablealgorithm}
     \kern2pt\hrule\relax% \@fs@post for \@fs@ruled
   \end{center}
  }
\makeatother

\graphicspath{ {./images/} }

\title{Enhanced Contour Tracking: a Time-Varying Internal Model Principle-Based Approach}

\author{
 Yue Cao \\
  Department of Mechanical Engineering\\
  Tsinghua University\\
  Beijing 100084, China\\
   \And
  Zhen Zhang* \\
  Department of Mechanical Engineering\\
  Tsinghua University\\
  Beijing 100084, China\\
  \texttt{zzhang@tsinghua.edu.cn} \\
}

\begin{document}
\maketitle
\begin{abstract}
Contour tracking plays a crucial role in multi-axis motion control systems, and it requires both multi-axial contouring as well as standard servo performance in each axis. Among the existing contouring control methods, the cross coupled control (CCC) lacks of an asymptotical tracking performance for general contours, and the task coordinate frame (TCF) control usually leads to system nonlinearity, and by design is not well-suited for multi-axis contour tracking. Here we propose a novel time-varying internal model principle-based contouring control (TV-IMCC) methodology to enhance contour tracking performance with both axial and contour error reduction. The proposed TV-IMCC is twofold, including an extended position domain framework with master-slave structures for contour regulation, and a time-varying internal model principle-based controller for each axial tracking precision improvement. Specifically, a novel signal conversion algorithm is proposed with the extended position domain framework, hence the original \boldmath$n$-axis contouring problem can be decoupled into (\boldmath$n-1$) two-axis master-slave tracking problems in the position domain, and the class of contour candidates can be extended as well. With this, the time-varying internal model principle-based control method is proposed to deal with the time-varying dynamics in the axial systems resulted from the transformation between the time and position domains. Furthermore, the stability analysis is provided for the closed-loop system of the TV-IMCC. Various simulation and experimental results validate the TV-IMCC with enhanced contour tracking performance compared with the existing  methods. Moreover, there is no strict requirement on the precision of the master axis, therefore a potential application of the TV-IMCC is multi-axis macro-micro motion systems.\\~\\
Index Terms—Precision motion control, contouring, contour tracking, internal model principle, time-varying systems, position domain, synchronized motion systems
\end{abstract}

% keywords can be removed

\section{Introduction}\label{intro}
Contour tracking, which is a popular and significant research direction in the field of precision mechatronics, has important applications to precision motion control for a scheduled two-/multi-axis contour, including machine tool processing~\cite{2li2020}, robot manipulation~\cite{wang2011}, and micro/nano machining~\cite{hu2011}. The contouring control has also emerged as an effective method to deal with macro-micro motion systems, such as the dual-drive stage system~\cite{yang2010} and the laser galvoscanner-stage system~\cite{cui2021}, to achieve both large working range and high precision. In the above applications, the contour to be tracked is composed of references in multiple motion dimensions, and the motion subsystems for all dimensions should be controlled simultaneously to ensure the contouring performance.

%In the past several decades, significant efforts have been devoted to the study of contour tracking, and numerous contour tracking control methods have been proposed and developed.
A basic idea to improve contour tracking precision is to achieve high tracking performance of each individual axis. And there are considerable advanced control methods available for axial precision tracking, such as adaptive control~\cite{farzan2022,fan2021}, sliding mode control (SMC)~\cite{qiao2019,boudjedir2022}, and iterative learning control (ILC)~\cite{tc2019,chen2021}. Nevertheless, it is worth noting that the improvement of axial tracking errors does not always imply reduced contour errors, because the contour error is a complex composition of all the axial individual errors. According to~\cite{tang2013}, the contour error can be defined as the minimum distance from the actual position to the reference contour. Clearly, the main task of a contouring method is the reduction or elimination of the contour error.

In order to tackle the above challenge, numerous contouring control methods have been proposed. First of all, the well-known cross-coupled control (CCC) was originally proposed in~\cite{koren1980,koren1991} to deal with basic contours such as lines, circles and their combinations. Thereafter, efforts have been made in~\cite{yeh2003,cheng2009} to extend the types of contours that can be dealt with CCC. Recently, the CCC and its modifications proposed in~\cite{barton2008,barton2011,2yang2010,uchiyama2011,ghaffari2015,zhang2018} (such as embedded ILC methods, modified error estimation methods, and model predictive methods) have been widely used due to the straightforward implementation. As the CCC lacks an asymptotical tracking performance for general contours, the contouring results may deteriorate in high-speed and/or large-curvature scenarios~\cite{ghaffari2015}.

Except of CCC, some other contour control strategies have been proposed to reduce the contour tracking errors in a different approach. One thread is the coordinate transformation-based method. For example, the tangential-contouring control (TCC)~\cite{lo1999} decouples the error dynamics into tangential errors and contour errors in the task coordinate frame (TCF), hence the control of the decoupled dynamics can be implemented separately. The TCF-based methods were modified in~\cite{chiu2001} for three-axis systems and for transformation matrix optimization in~\cite{cheng2007}. Furthermore, considering the nonlinearities caused by the TCF, the global task coordinate frame (GTCF) was developed in~\cite{yuan2016,hu2016,chen2019} to track the contours in relative high precision when combined with the adaptive control method. Nevertheless, the transformation matrix of the TCF methods becomes too complex to solve in the case of high dimensional motion axes, and the nonlinearities in the TCF setting restrict the potential axial controller integration for better tracking precision~\cite{tang2013}.

Another thread is the position domain control (PDC) method~\cite{ouyang2012}. In this method, the index variable of control is transformed from the temporal variable to the axial displacement variable. The PDC is able to completely decouple the contour error into the error dynamics of each individual axis in position domain. Hence, it is suitable for high dimensional contour situation~\cite{ouyang2013}. However, this method only utilizes PID as its axial controller, which means that the axial tracking performance may not be ensured even with a well-tuning process. And if the reference of the master axis (whose position is chosen as the position domain variable) is not monotonic with respect to the temporal variable $t$, stabilization can be only achieved piecewise by dividing the reference into sections to ensure monotonicity of each section~\cite{ouyang2012}.

Inspired by the existing contouring methods, we, in this paper, pursue further improvement in contour tracking performance by designing both contouring method and axial tracking method with enhanced precision. To achieve this goal, we first propose a novel and extended position domain framework, which ensures the multi-axis contouring performance by decoupling, and enables various classes of contours to be tracked. Then, based on the extended position domain framework, we find that the time-varying feature of the system dynamics is generated from the transformation between position domain and time domain. To tackle this issue, the time-varying internal model principle-based contouring control (TV-IMCC) is proposed to achieve asymptotic tracking for individual axis. It is known that he internal model principle (IMP) is a powerful tool for asymptotically tracking references and/or rejecting disturbances generated by autonomous systems~\cite{fran1977,isidori1990,serrani2001,galeani2011}. Meanwhile, some studies of the linear time-varying (LTV) IMP have been proposed in~\cite{zhang2009,zhang2010,zhang2014}. In this paper, the LTV IMP is deployed with the position domain framework to achieve improvement in both axial and multi-axial contouring precision.

The detailed contributions of this work are listed as follows:
\begin{itemize}
\item The multi-axis contouring control problem is transformed to decoupled axial time-varying tracking control by means of the proposed extended position domain framework, and the TV-IMCC is designed for each axis to achieve asymptotical tracking performance.
\item A novel signal conversion algorithm is designed to embed the rotational references of the master axis into the proposed extended position domain framework, and hence to enable the extension of the class of contour candidates.
\item A novel LMI-based time-varying stabilizer is designed by constructing a polytopic system state in the position domain for the proposed TV-IMCC.
\end{itemize}

The rest of the paper is organized as follows: in Section~\ref{prel}, the contouring system is described in the time domain and then decoupled in the position domain. To extend the classes of contours to be tracked, a novel position domain framework is designed in Section~\ref{frame}. The modeling of the master-slave contouring system is provided in Section~\ref{sec:slaveaxis}. The TV-IMCC including its stabilizer is proposed in Section~\ref{design}. Numerous simulation and experimental results are provided, validating the outperformed tracking results of the TV-IMCC in Section~\ref{exp}, followed by conclusion and future work in Section~\ref{conclusion}. For the convenience of the reader, we provide Fig.~\ref{guideline} to show a logical flow of the proposed TV-IMCC contouring method.
\begin{figure*}[htbp]
  \centering
    \includegraphics[width=\hsize]{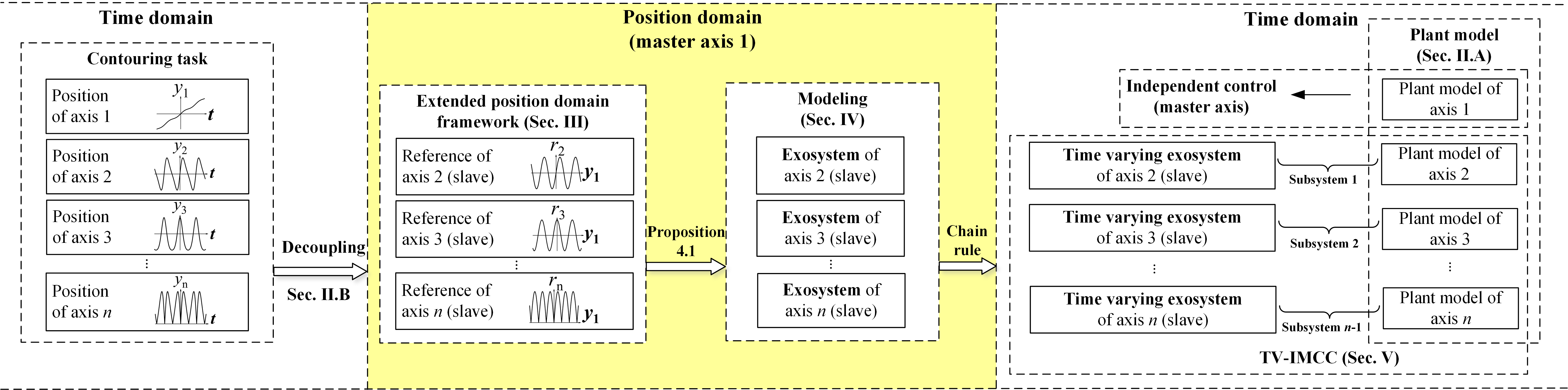}\\
  \caption{The logical flow of the proposed TV-IMCC method for multi-axis contouring.}\label{guideline}
\end{figure*}

\section{System description in time/position domain}\label{prel}
In this section, the precision contouring system is firstly described in the time domain. Then the position domain description of the contouring task is presented to decouple the original multi-axis system into multiple two-axis subsystems.
\subsection{System description in time domain}
\begin{figure}[!ht]
  \centering
  \includegraphics[width=0.7\hsize]{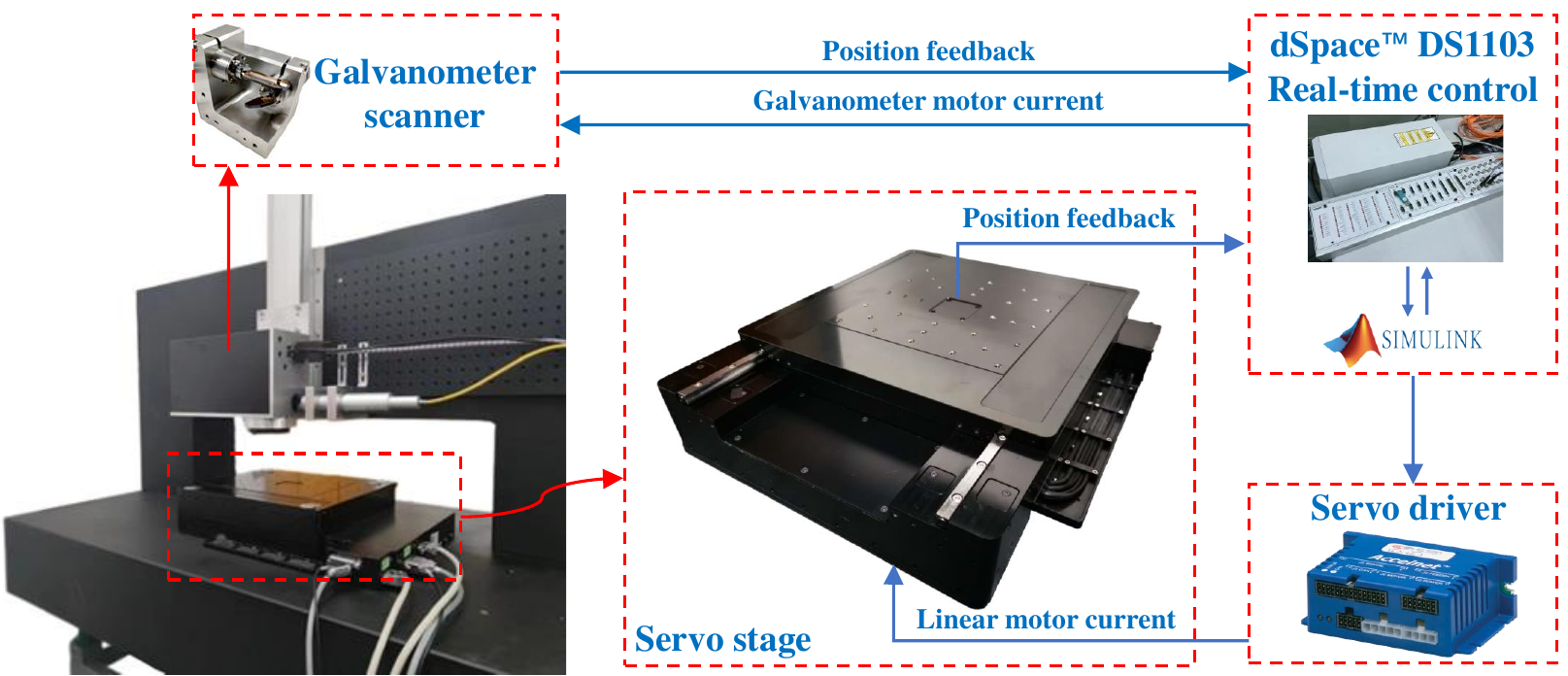}\\
  \caption{A self-developed galvoscanner-stage synchronized motion system.}\label{stage}
\end{figure}
The existing precision contouring tasks are implemented on a single motion system. A representative example is from the laser precision manufacturing industry, where either a motion stage or a galvanometer scanner (an electro-optical device utilizing a rotary low-inertia mirror to position a laser beam with high precision) is used for machining large-scale or delicate parts. Recently, there is a trend to combine multiple motion systems, so that the working area is enlarged by the motion stage, and the tracking precision is ensured with the scanner's fast and accurate motion. A straightforward way is the step-and-scan method, where the motion of stage and scanner is alternate. However, the step-and-scan method leads to stitching errors and low efficiency in manufacturing. More recently, the synchronized scan method is developed to enhance the scanning precision and efficiency, where the motion of stage and scanner is simultaneous (See the XL SCAN method of Scanlab$^\text{TM}$~\cite{scanlab}).

It is important to note that the synchronization of multiple motion systems leads to increased number of the axes. The existing contour tracking methods have difficulties in multi-axial contour error calculation and/or coordinate transformation, because the contouring task is considered as a coupled problem. Meanwhile, the PDC method is suitable for multi-axis contouring tasks, while the precision of its axial controller may not be satisfied. To tackle the above difficulties, we propose a completely decoupled contouring control methodology with asymptotic tracking performance.

To begin with, we consider a multi-axis contouring system, of which the nominal axial plant models are in the following form
\begin{equation}\label{si}
\begin{split}
\dot x_i(t)&=A_ix_i(t)+B_iu_i(t)\\[2mm]
y_i(t)&=C_ix_i(t)\,,
\end{split}
\end{equation}
where $i$ is the index of different axes; $x_i\in\R^n$ is the system state; $u_i\in\R$ is the control input; and $y_i\in\R$ is the position output.
\begin{assumption}
There is no coupling between each axis, which guarantees the existence of systems~(\ref{si}); the tuples $(A_i, B_i, C_i)$ are controllable and observable, and the systems~(\ref{si}) is asymptotically stable (if not, one can stabilize it).
\end{assumption}

In order to completely decouple the multi-axis contouring task, a decoupling method is designed in position domain, and hence the contouring control of the systems~(\ref{si}) can be generalized to a system with any number of motion axes.
\subsection{Decoupling multi-axis contouring task in position domain}\label{pddes}
%In the decoupling, the initial work is to decouple the multi-axial contouring control problem into the axial tracking control problem, so that the calculation can be significantly simplified.
The decoupling of the multi-axis contouring task is enabled via the transformation of the system dynamics from the time to position domain. We illustrate the multi-axis decoupling by starting from a two-axis case, where the contouring task can be completely decoupled into the independent master axis control, and the slave axis control in the master-axis-based position domain~\cite{ouyang2012}. Then, a multi-axis case can be similarly handled by dividing the multiple axes into one master axis and other slave ones,
\begin{align*}\label{multiaxis}
\begin{array}{rcl}
&&~~r_{\rm 1}=g_1(t)\,,\qquad\qquad\qquad\hspace{-.15mm} \rightarrow {\rm Master~axis}\\[1mm]
&&\left.
\begin{array}{c}
r_{\rm 2}=g_2(t)=f_2(y_{\rm 1})\,,\\[1mm]
r_{\rm 3}=g_3(t)=f_3(y_{\rm 1})\,,\\[1mm]
\vdots\\[1mm]
r_{\rm n}=g_{\rm n}(t)=f_{\rm n}(y_{\rm 1})\,,
\end{array}
\right\}~\rightarrow {\rm Slave~axes}
\end{array}
\end{align*}
where $n\in \N$ is the number of the axes, and $r_i~(i=1,2,\cdots,n)$ represent the axial contouring references, and $y_{\rm 1}$ represents the position of the master axis. The functions $g_i(\cdot)~(i=1,2,\cdots,n)$ represent the expressions of $r_i$ in time domain, while the functions $f_i(\cdot)~(i=2,3,\cdots,n)$ represent those in the position domain. Indeed, $r_i=f_i(y_1)$ is an implicit function of the two-variable function $F_i(y_1,\,r_i)=0$, which describes the characteristics of the contour, i.e., the relationship between the master axis and each slave axis.

Note that with the increased number of slave axes, the only burden of control is to apply the same tracking controller to the increased slave axis. Consequently, the multi-axis system can be divided into subsystems, and each of them has an individual slave axis, while all of them share the same master axis. More precisely, an $n$-dimensional contour tracking can be treated as a combination of $(n-1)$ numbers of two-dimensional contours. That is why the position domain framework is well suited for multi-axis contour tracking.

The decoupling is completed by modeling every two-axis subsystem in the position domain. Without loss of generality, we can mainly work on one two-axis contouring subsystem, containing one master axis (with subscript $1$) and one slave axis (with subscript $2$) for the ease of the notation burden and presentation in the rest of the paper.
\begin{remark}
\label{choose}
In a multi-axis motion system, it is preferred to choose the axis of the lowest motion precision as the master axis. It is because the control of the slave axis in position domain is capable of compensating the contour errors caused by the master axis.
\end{remark}

To extend the class of contours to be tracked, we study the expressions of contour signals in position domain, and design an extended position domain framework in what follows.
\section{Extended position domain framework}\label{frame}
Note that in the existing PDC~\cite{ouyang2012}, the position of the master axis is required to be a strictly monotonic function of the temporal variable $t$, otherwise the chain rule does not apply when the derivative is zero. However, a wide class of contours in industrial applications consists of non-monotonic rotational profiles (rectangles, circles, or irregular closed shapes with varying speed, acceleration and large curvature). Rotational signals are more general than periodic ones because they allow variations in every cycle of the contour. In fact, most of the rotational contour references can be mapped to monotonic functions with a virtual angular variable. Therefore, we extend the position domain framework, so that for a general two-dimensional contour, if at least one of the following assumptions is satisfied, the contour can be monotonized and then adopted into the extended position domain framework.
\begin{assumption}
\label{a1}
(Monotonic signals) The position of the master axis $y_1(t)$ is strictly monotonic with respect to the temporal variable $t$ ($\dot y_1(t)>0$ for all $t$, or  $\dot y_1(t)<0$ for all $t$), and the slave axis reference $r_2$ can be described as a differentiable function of $y_1$
\begin{equation}\label{a11}
r_2 = f(y_1)\,.
\end{equation}
\end{assumption}

An illustrative example of Assumption~\ref{a1} is given as $r_2=f(y_1)\triangleq \sin y_1$, which indicates a sinusoidal wave contour of amplitude $A=1$ and frequency $\omega=1$.
%which indicates a sinusoidal function of $x_1$.

\begin{assumption}
\label{a2}
(Rotational signals) The position of the master axis $y_1(t)$ can be transformed to a differentiable function of rotational variable $\theta(t)$ as
\begin{equation}\label{a21}
y_1(t)=(g_1\comp\theta)(t)\,,
\end{equation}
where the sign ($\comp$) denotes the composition of two functions, and $\theta(t)$ is strictly monotonic with respect to the temporal variable $t$. %, e.g., $\theta(t)\triangleq2t$, $g_1(\theta)\triangleq\cos \theta$, and $y_1(t)=(g_1\comp\theta)(t)=\cos 2t$.
The rotational variable $\theta(t)$ can be obtained through an inverse function of $y_1(t)$, that is
\begin{equation}\label{a22}
\theta(t)=g_1^{-1}(y_1(t))\,,
\end{equation}
%e.g., $\theta(t)=\cos^{-1}(\cos(2t))=2t$.
Further, the reference of slave axis  $r_2$ can be described as a differentiable function of $\theta$, i.e.
\begin{equation}\label{a23}
r_2=f(\theta)\,.
\end{equation}
\end{assumption}

An illustrative example of Assumption~\ref{a2} is given as $\theta(t)\triangleq t+0.5\sin t$, $g_1(\theta)\triangleq\cos \theta$, $y_1(t)=(g_1\comp\theta)(t)=\cos (t+0.5\sin t)$, and $r_2=f(\theta)\triangleq\sin\theta$. The shape of the contour ($y_1=\cos\theta=\cos (t+0.5\sin t)$, $r_2=\sin\theta=\sin (t+0.5\sin t)$) is a circle of radius $r=1$.

\begin{remark}
Note that the direct inverse function of $y_1(t)$ in equation~\eqn{a22} may not exist (e.g., $\cos^{-1}(\cdot)\neq \arccos(\cdot)$) for that the monotonicity of $y_1(t)$ does not hold, thus we need extra efforts to obtain $\theta(t)$ as follows.
\end{remark}

To solve function~(\ref{a22}) for the monotonic variable $\theta$ in the case of Assumption~\ref{a2}, we propose a convenient approach. Generally, variable $\theta$ can be recognized as the angular variable for a rotational signal, and the function~(\ref{a21}) maps $\theta$ to the actual position $y_1$. Note that the pair $(g_1,\theta)$ is not unique, here we choose the sinusoidal form of $g_1$ for calculation convenience. That is, $y_1$ can be described in the following form as
\begin{equation}\label{r1}
y_1(t)=R\cos\theta(t)
\end{equation}
in the continuous time domain, or with $t=kT_{\rm s}$ in the discrete-time domain,
\begin{equation}\label{rr1}
y_1(k)=R\cos\theta(k)\,,
\end{equation}
where $k$ is the sampling index, and $T_{\rm s}$ is the sampling period, and $R=\sup(\abs{y_1(t)})$. Then function~(\ref{a22}) can be solved by Algorithm~\ref{alg1}.
\begin{breakablealgorithm}
\caption{\small Conversion from a rotational variable $y_1$ to a monotonic angular variable $\theta_{\rm e}$}
\textbf{Initialize:} $s=0,~k=1$
\begin{algorithmic}[1]
\Repeat
    \State {$\theta_{\rm e}(k)=s\pi+\frac{\pi}{2}(1-(-1)^s)+(-1)^s\arccos (\frac{y_1(k)}{R})$}
    \If {$\theta_{\rm e}(k)<\theta_{\rm e}(k-1)$}
    \State {$s=s+1$}
    \State {$\theta_{\rm e}(k)=s\pi+\frac{\pi}{2}(1-(-1)^s)+(-1)^s\arccos (\frac{y_1(k)}{R})$}
    \EndIf
    \State \Return {$\theta_{\rm e}(k)$}
    \State {$k=k+1$}
\Until the reference signal ends
\end{algorithmic}
\label{alg1}
\end{breakablealgorithm}
%\begin{algorithm}[h]
%\label{alg1}
%\caption{\small Solving function~(\ref{a22})}
%\textbf{Initialize:} $s=0,~\theta_p=0$
%\begin{algorithmic}[1]
%\State {$\theta=s\pi+\frac{\pi}{2}(1-(-1)^s)+(-1)^s\arccos (\frac{x_m(k)}{R})$}
%\If {$(|x_m(k)-x_m(k-\delta)|<\varepsilon)$ and $(\theta-\theta_p>\sigma)$}
%\State {$s=s+1$}
%\State {$\theta=s\pi+\frac{\pi}{2}(1-(-1)^s)+(-1)^s\arccos (\frac{x_m(k)}{R})$}
%\State {$\theta_p=\theta$}
%\EndIf
%\State \Return {$\theta$}
%\end{algorithmic}
%\end{algorithm}
~\\[-5mm]
\begin{remark}
Algorithm~\ref{alg1} deals with the situation where $\theta(t)$ is strictly monotonically increasing. If $\theta(t)$ is strictly monotonically decreasing, the criterion of step 3 in Algorithm~\ref{alg1} will be $\theta_{\rm e}(k)>\theta_{\rm e}(k-1)$.% Without loss of generality, $\theta(t)$ is supposed to be monotonically increasing hereafter.
\end{remark}

To verify Algorithm~\ref{alg1}, the following lemma is provided.
\begin{lemma}\label{l1}
For $\alpha \in [s\pi,(s+1)\pi]$ ($s \in \N$), the following equation holds
\begin{equation}\label{ls1}
\alpha=s\pi+\frac{\pi}{2}(1-(-1)^s)+(-1)^s\arccos(\cos\alpha)\,.
\end{equation}
\end{lemma}

\begin{proof}
See Appendix~\ref{apx1}.
\end{proof}

With equation~(\ref{ls1}) in Lemma~\ref{l1}, we are now in position to present the following proposition to show that $\theta_{\rm e}$ is a feasible solution to equation~\eqn{a22}.

\begin{proposition}\label{prop:solutioneqn5}
If $\theta_{\rm e}(k)$ is obtained by Algorithm~\ref{alg1}, then
\begin{equation}\label{t1}
\theta_{\rm e}(k)=\theta(k)\,.
\end{equation}
\end{proposition}
\begin{proof}
See Appendix~\ref{apx:solutioneqn5}.
\end{proof}
%\begin{remark}
%%The ideal values of $\delta,~\varepsilon$ and $\sigma$ are $\delta=1,~\varepsilon \rightarrow 0$ and $\sigma \rightarrow \pi$. {\bf The practical values should be adjusted by experimental results. Not precise}
%\end{remark}

To illustrate Algorithm~\ref{alg1}, an example is provided in Fig.~\ref{alf}, where the rotational variable $y_1$ in triangle waveform is transformed to $\theta$ by Algorithm~\ref{alg1}, and  $\theta$ is strictly monotonic with respect to the temporal variable $t$.

%It is clear from the above analysis that the slave axis reference $r_2$ in the form~(\ref{a11}) or~(\ref{a23}) can be represented by a single independent variable ($y_1$ or $\theta$) which is the position of the master axis. This means that the position domain signal $f(y_1)$ (or $f(\theta)$) contains full information of the contour reference. Therefore, the better axial tracking precision is achieved for the slave axis for such a signal, the better contouring precision can be obtained, as long as the position of the master axis $y_1$ (or $\theta$) is bounded.%if a tracking controller is appropriately designed
\begin{figure}[!ht]
  \centering
  \includegraphics[width=0.5\hsize]{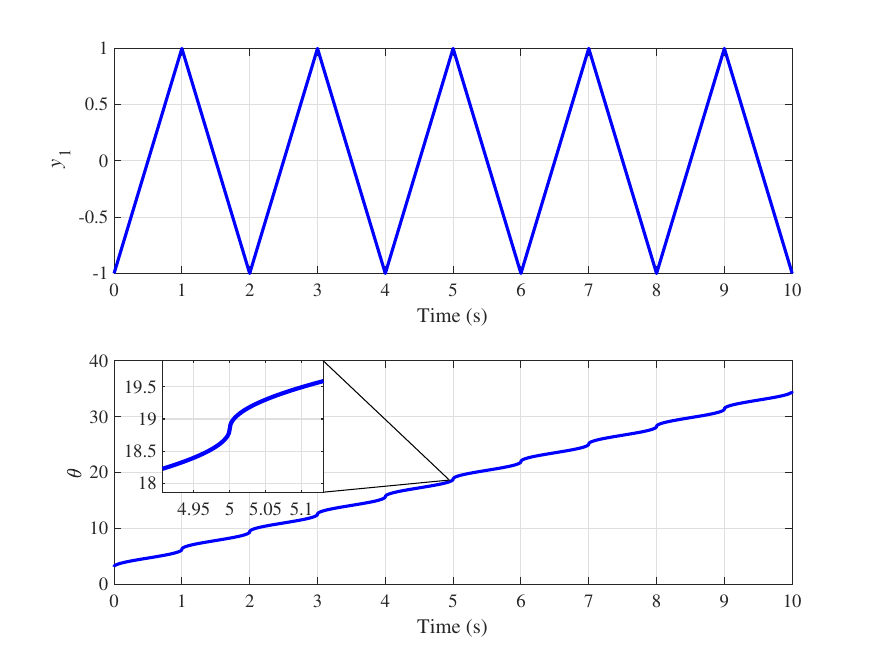}\\
  \caption{An illustrative example of Algorithm~\ref{alg1}.}\label{alf}
\end{figure}

We now summarize the features of the extended position domain framework by Proposition~\ref{pdcf}.%contours suitable for
\begin{proposition}\label{pdcf}
If Assumption~\ref{a1} or~\ref{a2} is satisfied, the contour reference dynamics can be analyzed in the extended position domain framework, namely, the derivative of the reference always exists in the position domain.
\end{proposition}
\begin{proof}
If Assumption~\ref{a1} is satisfied, the contour reference can be described in the form of equation~(\ref{a11}). Hence, the derivative of the reference in position domain can be written as
\begin{equation}\label{cr}
\dfrac{{\rm d}f(y_1)}{{\rm d}y_1}=\dfrac{{\rm d}f(y_1(t))}{{\rm d}t}\cdot \left[\dfrac{{\rm d}y_1(t)}{{\rm d}t}\right]^{-1}
\end{equation}
followed by the chain rule. Note that $y_1$ is strictly monotonic in time domain according to Assumption~\ref{a1}, that is
\begin{equation*}\label{dx1g0}
\dfrac{{\rm d}y_1(t)}{{\rm d}t}\neq0.
\end{equation*}
Therefore, the quantity $\left[\dot y_1(t)\right]^{-1}$ always exists, and so does the derivative with respect to the position variable $y_1$ in~equation~(\ref{cr}).

Likewise, if Assumption~\ref{a2} is satisfied, the proof follows by replacing $y_1$ with $\theta$. \hfill $\square$
\end{proof}

With the above proposition, it is ready to model the master-slave contouring system in the next section. Hereafter, the variable $\theta$ is omitted and only $y_1$ is used to represent the output of the master axis for simplicity.
\section{Modeling of master-slave system}\label{sec:slaveaxis}
In the master-slave contouring system, the control of the master axis is independent. Note that the tracking of the master axis can be adjusted to guarantee the range of $r_1$ and $y_1$ are the same. As a result, the contouring error of the master axis can be compensated by the control of the slave axis in the extended position domain framework to ensure asymptotic tracking. Therefore, we concentrate on the modeling and control of the slave axis, because the output of the master axis $y_1$ can be made bounded by any appropriate axial controller, for example PID.
\begin{remark}
In practice, the output of the master axis $y_1$ can be adjusted by scaling the control input or shifting the final time of control, so that the range of $r_1$ and $y_1$ is the same. Namely, the contouring error of the master axis can be completely compensated by the control of the slave axis.
\end{remark}

With the above discussion, it remains to obtain the generating dynamics (exosystem) and the plant model of the slave axis in appropriate domains.
\subsection{Generating dynamics of the master-slave system}\label{exo}
In the time domain, a class of signals to be tracked/rejected can be generated by an exosystem of the form
\begin{equation}\label{exotime}
\begin{array}{rcl}
\dot w(t)&=&S(t)w(t)\\[1mm]
f(t)&=&Q(t)w(t)\,,
\end{array}
\end{equation}
where $w\in\R^{\rho}$ is the exogenous state, $f(t)$ is the output of the exosystem. Note that the external disturbances with known generating dynamics can be also handled in (\ref{exotime}) by augmenting the disturbance dynamics into the exosystem of the internal model framework. And the detailed results are referred to~\cite{zhang2010,zhang2014}.

Nevertheless, in the TV-IMCC, the independent variable of the exogenous state is the position of the master axis $y_1$, which is illustrated in Assumptions~\ref{a1} and~\ref{a2}. Thus, the exosystem of the TV-IMCC is reconstructed in the position domain as
\begin{equation}\label{exop}
\begin{array}{rcl}
w'(y_1)&=&S(y_1)w(y_1)\\[1mm]
f(y_1)&=&Q(y_1)w(y_1)\,,
\end{array}
\end{equation}
where
\[
w'(y_1)\triangleq\dfrac{{\rm d}w(y_1)}{{\rm d}y_1}\,,
\]
and the sign ($'$) is used to represent the position domain derivatives throughout the paper. The exosystem~(\ref{exop}) should contain the information of the contour reference to be tracked. Therefore, the parameters of the exosystem are identified by the position of the master axis $y_1$ and the reference of the slave axis $r_2 = f(y_1)$ in the following proposition.
\begin{proposition}\label{prop:exo}
The state space model of the exosystem $S(y_1)$ and $Q(y_1)$ can be described in the following form
\begin{equation}\label{exomodel}
\begin{array}{rcl}
S(y_1)&=&\begin{bmatrix} \dfrac{l'(y_1)}{l(y_1)} & -\eta'(y_1)\\\eta'(y_1) &\dfrac{l'(y_1)}{l(y_1)}\end{bmatrix}\\[10mm]
Q(y_1)&=&\begin{bmatrix} 1& 0 \end{bmatrix}\,,
\end{array}
\end{equation}
where $l(y_1)$ is defined as
\begin{equation}\label{l}
l(y_1)\triangleq\sqrt{y_1^2+f^2(y_1)}\,,
\end{equation}
and $\eta(y_1)$ is defined as
\begin{equation}\label{eta}
\eta(y_1) \triangleq \arccos\dfrac{f(y_1)}{\sqrt{y_1^2+f^2(y_1)}}\,.
\end{equation}
\end{proposition}
\begin{proof}
See Appendix~\ref{apx:exo}.
\end{proof}
\subsection{Problem formulation of the slave axis control}\label{prob}
With the constructed exosystem~(\ref{exop}) in the position domain, it is noted that the plant model of the slave axis is in the time domain according to~(\ref{si}). In order to keep the consistency of the domains, the exosystem~(\ref{exop}) is transformed into the time domain via the chain rule, and the time-varying parameters are introduced accordingly. As a result, the concerned system of the slave axis can be written in the time domain as follows
\begin{align}
\begin{split}
\label{plant}
\hspace{15mm}\dot x_2(t)&~=~A_2x_2(t)+B_2u_2(t)\,\\[1mm]
\hspace{15mm}y_2(t)&~=~C_2 x_2(t)\,\\[1mm]
\hspace{15mm}e_2(t)&~=~y_2(t)-f(t)\,,
\end{split}&
\\[2mm]
\begin{split}
\label{exoplantime}
\hspace{15mm}\dot w_y(t)&~=~S_{y}(t)w_y(t)\,\\[1mm]
\hspace{15mm}f_y(t)&~=~Q_{y}(t)w_y(t)\,,
\end{split}&
\end{align}
with system state $x_2\in\R^n$, control input $u_2\in\R$, output $y_2\in\R$, contour error $e_2\in\R$, reference $f\in\R$. The matrices $S(y_1)$ and $Q(y_1)$ are obtained via the chain rule as follows
\begin{equation}\label{SQtime}
\begin{array}{rcl}
S_{y}(t)&=&\dot y_1(t)S(y_1)=\begin{bmatrix}\dfrac{\dot l(y_1(t))}{l(y_1(t))} & -\dot \eta(y_1(t))\\ \dot \eta(y_1(t)) & \dfrac{\dot l(y_1(t))}{l(y_1(t))}\end{bmatrix}\,\\[10mm]
Q_{y}(t)&=&Q(y_1)=\begin{bmatrix} 1&0\end{bmatrix}\,,
\end{array}
\end{equation}
and $f_y(t)$ and $w_y(t)$ stand for the composite functions $(f\comp y_1)(t)$ and $(w\comp y_1)(t)$, respectively.
%\[
%\dot w_y(t)=\dfrac{{\rm d}w(y_1)}{{\rm d}y_1}\cdot\dfrac{{\rm d}y_1}{{\rm d}t}=w'(y_1)\cdot \dot y_1(t)\,,
%\]
%\[
%\dot \eta(y_1(t))=\dfrac{{\rm d}\eta(y_1)}{{\rm d}y_1}\cdot\dfrac{{\rm d}y_1}{{\rm d}t}=\eta'(y_1)\cdot \dot y_1(t)\,,
%\]
%and
%\[
%\dot l(y_1(t))=\dfrac{{\rm d}l(y_1)}{{\rm d}y_1}\cdot\dfrac{{\rm d}y_1}{{\rm d}t}=l'(y_1)\cdot \dot y_1(t)\,.
%\]

For the sake of digital control (Algorithm~\ref{alg1} is in a discrete version), the system model of interest is discretized from the continuous one~(\ref{plant})-(\ref{exoplantime}) as follows
\begin{align}
\begin{split}
\label{dplant}
\hspace{15mm}x_2(k+1)&=G_2x_2(k)+H_2u_2(k)\\[1mm]
\hspace{15mm}y_2(k)&=C_2x_2(k)\\[1mm]
\hspace{15mm}e_2(k)&=y_2(k)-f(k)\,,
\end{split}&
\\[2mm]
\begin{split}
\label{dexoplant}
\hspace{15mm}w(k+1)&=\bar S(k)w(k)\\[1mm]
\hspace{15mm}f(k)&=\bar Q(k)w(k)\,,
\end{split}&
\end{align}
where $t=kT_s$, and $k$ is the sampling index, and $T_s$ is the sampling period. Specifically, by noting that the commutative law of multiplication is satisfied for $S_y(t)$ in~(\ref{SQtime}), the discrete-time form $\bar S(k)$ can be calculated by $\bar S(k)={\rm exp}\left(\int_{kT_s}^{(k+1)T_s}S_y(\tau){\rm d}\tau\right)$.

Hereafter, the original contouring problem described in Section~\ref{pddes} is transformed to axial tracking problem for~(\ref{dplant})-(\ref{dexoplant}) by means of finding a compensator fed only by the measurement of the tracking error $e_2(k)$, satisfying
\[
\lim_{k\to\infty}e_2(k)=0\,.
\]

It is important to note that $\bar S(\cdot)$ is time-varying, hence we develop a time-varying internal model principle-based approach to solve the above tracking problem, including the time-varying internal model in Section~\ref{controller}, and the time-varying stabilizer in Section~\ref{stabilizer}.
\section{Design of the TV-IMCC}\label{design}
\begin{figure*}[htbp]
  \centering
  \includegraphics[width=0.8\hsize]{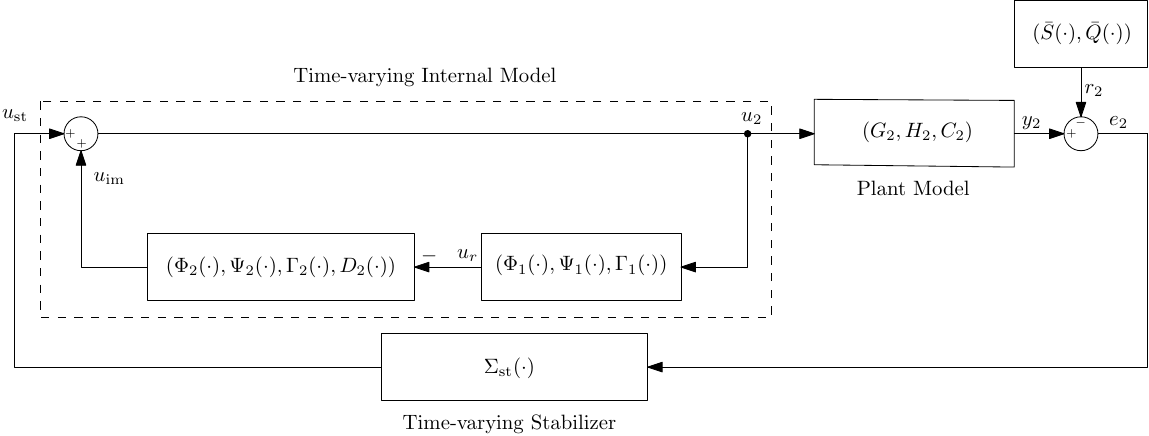}\\
  \caption{The block diagram of the TV-IMCC in a slave axis.}\label{bd}
\end{figure*}
The structure of the TV-IMCC is shown in Fig.~\ref{bd}, where  the controller  is composed of the time-varying internal model and stabilizer.
\subsection{Time-varying internal model design}\label{controller}
The tracking problem for systems~(\ref{dplant})-(\ref{dexoplant}) is solved by finding an error-feedback compensator of the form
\begin{align}
\label{compensator}
\begin{split}
\zeta(k+1)&= M(k)\zeta(k)+N(k) e_2(k)\\[1mm]
u_2(k)&=J(k)\zeta(k)+L(k)e_2(k)\,,
\end{split}
\end{align}
where $\zeta\in\R^{n_\zeta}$ is compensator state, and $M(k), N(k), J(k)$, and $L(k)$ are matrices with appropriate dimensions. The existence of~(\ref{compensator}) is equivalent to the solvability of the following time-varying difference equations (see a continuous-time version in~\cite{zhang2010})
\begin{align}
\begin{split}
\Pi(k+1)\bar S(k)&= G_2\Pi(k)+H_2R(k)\\[1mm]
0&=C_2\Pi(k)+\bar Q(k)\,,
\end{split}
\end{align}
\begin{align}\label{impm}
\begin{split}
\Sigma(k+1)\bar S(k)&= M(k)\Sigma(k)\\[1mm]
R(k)&=J(k)\Sigma(k)\,,
\end{split}
\end{align}
where $\Pi:\R\mapsto\R^{n\times\rho}$, $\Sigma:\R\mapsto\R^{n_\zeta\times\rho}$ and $R:\R\mapsto\R^{1\times\rho}$ are smooth mappings. Eq.~(\ref{impm}) constitutes the counterpart in the LTV setting of the internal model principle. According to the internal model principle~\cite{fran1977}, a desired input $u_{\rm d}(k)=R(k)w(k)$ is an unavailable feedforward term of $u_2(k)$ that keeps the error $e_2(k)$ identically at zero. As a result, one needs to reconstruct $u_{\rm d}$ as an output of the system $(\bar S(k),\ R(k))$, that is
\begin{align}
\label{desire}
\begin{split}
w(k+1)&=\bar S(k)w(k)\\[1mm]
u_{\rm d}(k)&=R(k)w(k)\,.
\end{split}
\end{align}
Without directly solving $R(k)$, we construct $u_{\rm d}$ by resorting to the concept of \emph{system immersion}~\cite{zhang2009}. Namely, every output of~(\ref{desire}) is an output of another system of the form,
\begin{align}
\label{immer}
\begin{split}
\xi(k+1)&=\Phi(k)\xi(k)\\[1mm]
\eta(k)&=\Gamma(k)w(k)\,,
\end{split}
\end{align}
where $\Phi(k)$, and $\Gamma(k)$ are matrices of appropriate dimensions to be designed, and $\xi\in\R^{n_\xi}$ is the associated state, and $\eta(k)\in \R$ is the required  input. Note that there is no general method to directly find such an immersion for LTV systems~\cite{zhang2009,zhang2010}. As an alternative, we construct~(\ref{immer}) by two internal model units, i.e.
\begin{equation}\label{ims1}
\begin{array}{rcl}
\xi_1(k+1)&=&\Phi_1(k)\xi_1(k)+\Psi_1(k)u_2(k)\\[1mm]
u_r(k)&=&\Gamma_1(k)\xi_1(k)\,,
\end{array}
\end{equation}
and
\begin{equation}\label{ims2}
\begin{array}{rcl}
\xi_2(k+1)&=&\Phi_2(k)\xi_2(k)+\Psi_2(k)(-u_r(k))\\[1mm]
u_{\rm im}(k)&=&\Gamma_2(k)\xi_2(k)+D_2(k)(-u_r(k))\,,
\end{array}
\end{equation}
where $\Phi_1(k)$, $\Psi_1(k)$, $\Gamma_1(k)$, $\Phi_2(k)$, $\Psi_2(k)$, $\Gamma_2(k)$, and $D_2(k)$ are matrices to be designed of appropriate dimensions, and $(\xi_1,\,\xi_2)\in(\R^{\rho},\,\R^{\rho-1})$ is the state, and $u_r\in\R$ is the embedded output, and $u_{\rm im}\in\R$ is the internal model input. The motivation of this design is to generate the reference $r_2$ in the place of $u_r$, i.e., $u_r=r_2$. By producing a self-excited input $u_2=u_{\rm im}$, systems~(\ref{ims1})-(\ref{ims2}) can be written as follows,
\begin{align}\label{aux}
\begin{split}
\begin{bmatrix}
\xi_2(k+1)\\ \xi_1(k+1)
\end{bmatrix} =
&\begin{bmatrix}
\Phi_2(k) & -\Psi_2(k)\Gamma_2(k)\\ \Psi_2(k)\Gamma_2(k)&\Phi_2(k)- \Psi_2(k)D_2(k)\Gamma_2(k)
\end{bmatrix}
\begin{bmatrix}
\xi_2(k)\\ \xi_1(k)
\end{bmatrix}\\[2mm]
v=&
\begin{bmatrix}
\Gamma_2(k) & 0
\end{bmatrix}
 \begin{bmatrix}
\xi_2(k)\\ \xi_1(k)
\end{bmatrix}\,,
\end{split}
\end{align}
with an auxiliary output $v\in \R$. Note that the exosystem~(\ref{dexoplant}) can be immersed into the above system~(\ref{aux}) by noting Lemma~\ref{l2}.
\begin{lemma}\label{l2}
Let $(\Phi_1(\cdot),\Psi_1(\cdot),\Gamma_1(\cdot))$ and $(\bar S(\cdot), \bar Q(\cdot))$ be in the observer canonical form, and $(\Phi_2(\cdot),\Psi_2(\cdot),\Gamma_2(\cdot))$ be in the phase-variable (controller) canonical form. And there exists a unique solution to the following Sylvester equation
\begin{equation}\label{se}
\begin{bmatrix}\mathcal{O}_{\Phi_1(\cdot)} & \mathcal{C}_{\Psi_1(\cdot} \end{bmatrix} \begin{bmatrix}1\\q(\cdot)\\p(\cdot)\end{bmatrix}= \mathcal{O}_{\bar S(\cdot)}\begin{bmatrix}1\\q(\cdot)\end{bmatrix}\,,
\end{equation}
such that the exosystem~(\ref{dexoplant}) is immersed into system~(\ref{aux}).
\end{lemma}
\begin{proof}
Note that the pair $(\Phi_1(\cdot)-\bar S(\cdot)\Gamma_1^\top(\cdot)\Gamma_1(\cdot),\ \Gamma_1(\cdot))$ is uniformly observable by canonical form design.
Then the solvability of~(\ref{se}) and the proof are referred to our previous results in Lemma~3.2 and Proposition~3.3 of Ref.\cite{zhang2010}. \hfill $\square$
\end{proof}

The definition of the canonical forms in Lemma~\ref{l2} is referred to~\cite{silverman1966}. In the discrete-time setting of Lemma~\ref{l2}, the variables $q(k)\in\R^{\rho-1}$ and $p(k)\in\R^{\rho-1}$ are the parameter vectors of $\Phi_2(k)$ and $\Gamma_2(k)$ in the controllability canonical form respectively, and $\Psi_2(k)=\begin{bmatrix}0 & 0& \cdots & 1\end{bmatrix}^{\rm T},~D_2(k)=p_{\rho-1}(k)$. The operator $\mathcal{O}_{\bar S(k)}\in\R^{(2\rho-1)\times\rho}$ is defined as
\begin{equation}\label{os}
\mathcal{O}_{\bar S(k)}=\begin{bmatrix} \mathcal{G}_{\rho-1} & \cdots & \begin{bmatrix} 0_{\rho-3} \\ 1 \\ \mathcal{G}_1 \end{bmatrix} & \begin{bmatrix} 0_{\rho-2} \\ 1 \\ \mathcal{G}_0 \end{bmatrix} \end{bmatrix} \,,
\end{equation}
with $\alpha(k)\in\R^{\rho}$ collecting the coefficients of the exosystem dynamics matrix $\bar S(k)$ in its observer canonical form $\bar S_o(k)$. And the operator $\mathcal{G}$ is defined as
\begin{equation*}
\mathcal{G}_{i+1}= \mathcal{q}^{-1}\begin{bmatrix} \mathcal{G}_i \\ 0 \end{bmatrix},\quad \mathcal{G}_0=\alpha(k+\rho-1),\quad i=0,\cdots,\rho-2,
\end{equation*}
where $\mathcal{q}^{-1}$ is the unit delay operator. Another operator $\mathcal{O}_{\Phi_1(k)}\in\R^{(2\rho-1)\times\rho}$ is defined similarly with $\mathcal{O}_{\bar S(k)}$. The operator $\mathcal{C}_{\Psi_1(k)}\in\R^{(2\rho-1)\times\rho}$ is defined as
\begin{equation}\label{cp}
\mathcal{C}_{\Psi_1(k)}=\begin{bmatrix} \mathcal{H}_{\rho-1} & \cdots & \begin{bmatrix} 0_{\rho-2} \\ \mathcal{H}_1 \end{bmatrix} & \begin{bmatrix} 0_{\rho-1} \\ \mathcal{H}_0 \end{bmatrix} \end{bmatrix} \,,
\end{equation}
with
\begin{equation*}
\mathcal{H}_{i+1}= \mathcal{q}^{-1}\begin{bmatrix} \mathcal{H}_i \\ 0 \end{bmatrix},\quad \mathcal{H}_0=\Psi_1(k+\rho-1),\quad i=0,\cdots,\rho-2.
\end{equation*}
Note that an exponential stable LTI filter can be utilized to augment the exosystem~\cite{zhang2010}, such that dim$(\xi_1)$ and dim$(x_2)$ are the same ($n=\rho$), and hence both sides of~(\ref{se}) have the same dimension.

Given Lemma~\ref{l2}, to immerse the desired system~(\ref{desire}) into system~(\ref{immer}), the tuple $(\Phi_{1}(\cdot),~\Psi_{1}(\cdot),~\Gamma_{1}(\cdot))$ should have the same I/O (input/output) form as that of the plant model according to Lemma 4.1 in~\cite{zhang2010}. Let the plant model $(G_{2},~H_{2},~C_{2})$ in the observer canonical form, the easiest design of $(\Phi_{1}(k),~\Psi_{1}(k),~\Gamma_{1}(k))$ in~(\ref{ims1}) can be chosen as
\begin{equation}\label{ims1d}
(\Phi_{1}(k),~\Psi_{1}(k),~\Gamma_{1}(k))=(G_{2},~H_{2},~C_{2})\,.
\end{equation}

\begin{remark}
The tuple $(\Phi_{1}(k),~\Psi_{1}(k),~\Gamma_{1}(k))$ is replaced by $(G_{2},~H_{2},~C_{2})$ hereafter.
\end{remark}
%Suppose that the plant model is stable. With the above design~(\ref{ims1})-(\ref{ims1d}), it is left to stabilize the following augmented system
%\begin{equation}\label{as}
%\begin{scriptsize}
%\begin{array}{rcl}
%\begin{bmatrix} \xi_2(k+1)\\x_2(k+1)\end{bmatrix}&=&\begin{bmatrix} \Phi_2(k)& -\Psi_2(k)C_2(k) \\ H_2(k)\Gamma_2(k) & G_2(k)-H_2(k)D_2(k)C_2(k)\end{bmatrix}\begin{bmatrix} \xi_2(k)\\x_2(k)\end{bmatrix}\\[3mm]
%&&+\begin{bmatrix} 0\\H_2(k)\end{bmatrix}u_{st}(k)\,.
%\end{array}
%\end{scriptsize}
%\end{equation}
\subsection{Time-varying stabilizer design}\label{stabilizer}
%\begin{figure*}[!ht]
%\subfloat{
%\includegraphics[width=0.45\linewidth]{la.pdf}}
%\hfill
%\hspace{.4in}
%\subfloat{
%\includegraphics[width=0.45\linewidth]{ua.pdf}}
%\hfill
%\caption{Frequency response of the experimental setup. Left: $X_1$; Right: $X_2$.}
%\label{fr}
%\end{figure*}
With the above time-varying internal model~(\ref{ims1})-(\ref{ims1d}), it remains to design a time-varying stabilizer for the augmented system~(\ref{augment}) of the plant model and the time-varying internal model with $r_2=f(k)=0$ as the reference plays no role in stabilization. \begin{align}\label{augment}
\begin{split}
\begin{bmatrix}\xi_2(k+1)\\x_2(k+1) \end{bmatrix} =&\begin{bmatrix}\Phi_2(k) & -\Psi_2(k)C_2\\H_2\Gamma_2(k)&G_2- H_2D_2(k)C_2 \end{bmatrix} \begin{bmatrix}\xi_2(k)\\x_2(k) \end{bmatrix} \\[3mm]
&+ \begin{bmatrix}0\\H_2 \end{bmatrix}u_{\rm st}(k)\,.
\end{split}
\end{align}
\begin{lemma}\label{stabilization}
To obtain the stabilizer input $u_{\rm st}\in \R$, we first provide the following lemma on stabilization of the system~(\ref{augment}).

A sufficient condition of stabilizing the  augmented system~(\ref{augment}) is to stabilize the following one
\begin{equation}\label{augsys}
\begin{array}{rcl}
x_{\rm o}(k+1)&=&A(k)x_{\rm o}(k)+B(k)u_{\rm st}(k)\,\\[1mm]
y_2(k)&=&\begin{bmatrix}1&0&\cdots&0\end{bmatrix}x_{\rm o}(k)\,,
\end{array}
\end{equation}
where
\begin{equation*}
A(k)=\begin{bmatrix} -\mathcal{Q}\alpha(k)_{\rho\times1} & I_{(n-1)\times(n-1)}\\0_{(n-\rho)\times1}&0_{1\times(n-1)} \end{bmatrix}\,,\,B(k)=H_2\,.
\end{equation*}
And the operator $\mathcal{Q}$ is defined as
\begin{equation*}
\mathcal{Q}=\rm{diag}(\mathcal{q}^{\rho-1},\mathcal{q}^{\rho-2},\cdots,1).
\end{equation*}
\end{lemma}
\begin{proof}
The proof is referred to the previous work in Lemma 3.1 of Ref.\cite{zhang2014}. \hfill $\square$
\end{proof}
Note that system~(\ref{augsys}) can be split as
\begin{equation}\label{splitsys}
\begin{array}{rcl}
x_{\rm o1}(k+1)&=&A_{11}(k)x_{\rm o1}(k)+A_{12}(k)x_{\rm o1}(k)+B_1(k)u_{\rm st}(k)\,\\[1mm]
x_{\rm b}(k+1)&=&A_{21}(k)x_{\rm b}(k)+A_{22}(k)x_{\rm b}(k)+B_2(k)u_{\rm st}(k)\,\\[1mm]
y_2(k)&=&x_{\rm o1}(k)\,,
\end{array}
\end{equation}
where $x_{\rm o1}\in\R$ is the first state of $x_{\rm o}$, and $x_{\rm b}\in\R^{n-1}$ collects the rest $(n-1)$ states of $x_{\rm o}$. %, and $A_{11}(k),A_{21}(k),A_{12},A_{22},B_1(k),B_2(k)$ are the corresponding splits of $A(k)$ and $B(k)$.
We introduce a reduced order observer of the state $x_{\rm b}$ as follows:
\begin{equation}\label{ob}
\begin{array}{rcl}
\hat z(k+1)&=&(A_{22}-FA_{12})\hat z(k)+(B_2(k)-FB_1(k))u_{\rm st}(k)\\[1mm]
&&+((A_{22}-FA_{12})F+A_{21}(k)-FA_{11}(k))y_2(k)\,,\\[1mm]
\hat x_{\rm b}(k)&=&\hat z(k)+Fy_2(k)\,,
\end{array}
\end{equation}
where $F$ is the output injection gain of the observer. Here $F$ can be chosen as a constant vector for that the pair $(A_{21},A_{22})$ is time-invariant according to the form of $A(k)$ in Lemma~\ref{stabilization}. With the estimation $\hat x_{\rm b}(k)$, we can stabilize system~(\ref{splitsys}) via state feedback, with the stabilizer input
\begin{equation*}
u_{\rm st}=\begin{bmatrix}K_1(k) & K_2(k)\end{bmatrix}\begin{bmatrix}x_{\rm o1}(k)\\ \hat x_{\rm b}(k)\end{bmatrix},
\end{equation*}
where $K(k)=\begin{bmatrix}K_1(k) & K_2(k)\end{bmatrix}\in \begin{bmatrix}\R & \R^{1\times(n-1)}\end{bmatrix}$ is the feedback gain. Hence, the remaining task is to make the closed-loop system
\begin{equation}\label{stabsys}
x_{\rm o}(k+1)=(A(k)+B(k)K(k))x_{\rm o}(k)
\end{equation}
asymptotically stable. To begin with, we show that the time-varying matrix $A(k)$ can be split by polytope in the following Proposition~\ref{a3}.
\begin{proposition}
\label{a3}
The time-varying terms of $A(k)$ are parameter $\sigma (k)$ dependent and $\sigma(k)\in \R$ belongs to a polytope, that is
\begin{equation}\label{poly}
A(k)=A(\sigma(k))=\sum\limits_{i=1}^N\sigma_i(k)A_i\,,
\end{equation}
where
\begin{equation*}
\sigma_i(k)\ge0,~~\sum\limits_{i=1}^N\sigma_i(k)=1\,,
\end{equation*}
and $A_i$'s are constant matrices.
\end{proposition}
\begin{proof}
From~equation (\ref{SQtime}), it is clear that the time-varying terms of $A(k)$ are parameter $y_1(k)$ dependent. Also, $y_1(k)$ is strictly monotonic according to Assumption~\ref{a1} and~\ref{a2}, thus it is ensured that $y_1(k)> 0$ by adding a constant term if necessary. Therefore, the parameter $\sigma (k)$ in equation~\eqn{poly} can be designed as a function of $y_1(k)$.

Note that $y_1(k)$ is bounded due to the independent tracking of the master axis in the TV-IMCC framework. Thus, one can discretize $y_1$ into $N$ gridding points as
\begin{equation*}
\begin{split}
&y_1^i\in y_1^1,y_1^2,\dots,y_1^N\\
&y_1^1=\inf y_{1}\\
&y_1^N=\sup y_{1}\,,
\end{split}
\end{equation*}
and
\begin{equation*}
A_i=A(y_1^i)\,.
\end{equation*}
Therefore for each $y_1(k)$, parameter $\sigma_i(k)$ is defined by interpolation as follows
\begin{equation}\label{sigma}
\sigma_i(k)=\prod_{j \neq i,~j\in 1,2,\dots,N}\dfrac{y_1(k)-y_1^j}{y_1^i-y_1^j}\,,
\end{equation}
which satisfies equation~\eqn{poly}.~\hfill $\square$
\end{proof}

Meanwhile, note that $B(k)$ is time-invariant in our design. Therefore, the feedback gain $K(k)$ can be obtained by solving the following linear matrix inequalities (LMIs) in Lemma~\ref{polytype}.
\begin{lemma}\label{polytype}
The stabilizer gain $K(k)$ of system~(\ref{stabsys}) can be obtained if the following LMIs are solvable for all $i=1,\,2,\cdots,N$ and $j=1,\,2,\cdots,N,$
\begin{equation}\label{lmi}
\begin{bmatrix} G_i+G_i^\top-Q_i & (A_iG_i+BR_i)^\top\\A_iG_i+BR_i & Q_j\end{bmatrix}>0\,,
\end{equation}
where $Q_i,~Q_j\in \R^{n\times n}$ (symmetric positive-definite matrices), and $G_i\in \R^{n\times n},~R_i\in \R^{1\times n}$ are the solutions of the above LMIs. Then the stabilizer gain $K(k)$ is obtained by
\begin{equation}\label{k}
K(k)=\sum_{i=1}^N\sigma_i(k)R_iG_i^{-1}\,.
\end{equation}
\end{lemma}
\begin{proof}
See Appendix~\ref{apx2}.
\end{proof}

With the internal model controller in Section~\ref{controller} and the stabilizer in Section~\ref{stabilizer}, we are in position to state the TV-IMCC controller.
\subsection{TV-IMCC controller design}\label{tvcontroller}
The TV-IMCC controller is provided as the following result.
\begin{theorem}
The tracking error $e_2$ of systems~(\ref{dplant})-(\ref{dexoplant}) converges to zero with the following TV-IMCC controller
\begin{equation}\label{stab}
%\begin{scriptsize}
\begin{array}{rcl}
\hspace{-2mm}\xi_1(k+1)&\hspace{-3mm}=&\hspace{-3mm}G_2\xi_1(k)+H_2u_2(k)\\[1mm]
\hspace{-2mm}u_r(k)&\hspace{-3mm}=&\hspace{-3mm}C_2\xi_1(k)\\[3mm]
\hspace{-2mm}\xi_2(k+1)&\hspace{-3mm}=&\hspace{-3mm}\Phi_2(k)\xi_2(k)+\Psi_2(k)(-u_r(k))\\[1mm]
\hspace{-2mm}u_{\rm im}(k)&\hspace{-3mm}=&\hspace{-3mm}\Gamma_2(k)\xi_2(k)+D_2(k)(-u_r(k))\\[3mm]
\hspace{-2mm}\hat z(k+1)&\hspace{-3mm}=&\hspace{-3mm}(A_{22}(k)-FA_{12}(k))\hat z(k)+(B_2-FB_1)u_{\rm st}(k)+\\[1mm]
\hspace{-2mm}&\hspace{-3mm}&\hspace{-3mm}((A_{22}(k)-FA_{12}(k))F+A_{21}(k)-FA_{11}(k))e_2(k)\\[1mm]
\hspace{-2mm}\hat x_{\rm b}(k)&\hspace{-3mm}=&\hspace{-3mm}\hat z(k)+Fe_2(k)\\[1mm]
\hspace{-2mm}u_{\rm st}(k)&\hspace{-3mm}=&\hspace{-3mm}K_1(k)e_2(k)+K_2(k)\hat x_{\rm b}(k)\\[3mm]
\hspace{-2mm}u_2(k)&\hspace{-3mm}=&\hspace{-3mm}u_{\rm im}(k)+u_{\rm st}(k)\,,
\end{array}
%\end{scriptsize}
\end{equation}
where $(\Phi_{2}(k),~\Psi_{2}(k),~\Gamma_{2}(k),~D_{2}(k))$ is obtained by solving the Sylvester equation~(\ref{se}); and $K(k)$ is obtained in~(\ref{k}) by solving the LMIs~(\ref{lmi}).
\end{theorem}
\begin{proof}
The proof immediately follows by noting Lemma~\ref{l2}, \ref{stabilization}, \ref{polytype}, Proposition~\ref{a3}, the I/O mapping~(\ref{ims1d}), and the reduced order observer~(\ref{ob}).\hfill $\square$
\end{proof}
\section{Simulation and experimental study}\label{exp}
In this section, numerous simulations and experiments are conducted to validate the proposed TV-IMCC methodology. We start from the biaxial contours tracked by the XY servo stage shown in Fig.~\ref{stage}. The discrete model of the two axes (master: $X_1$; slave: $X_2$) is identified with sampling rate $1$kHz. The parameters of the master axis model $(G_1,H_1,C_1)$ and the slave axis model $(G_2,H_2,C_2)$ are identified as follows
\begin{equation}
\begin{split}
&G_1=[1.9734\quad1;\quad-0.9735\quad 0],\\
&H_1=[2.5259{\rm e}^{-4}\quad2.5034{\rm e}^{-4}]^\top,\quad C_1=[1\quad0];\\
&G_2=[1.9581\quad1;\quad-0.9583\quad 0],\\
&H_2=[6.8214{\rm e}^{-4}\quad6.7253{\rm e}^{-4}]^\top,\quad C_2=[1\quad0].
\end{split}
\end{equation}

We make comparisons among the tracking results of axial PID, CCC, TCF and the proposed TV-IMCC. The axial PID gains are tuned as $K_{\rm p_1}=34.96$, $K_{\rm i_1}=173.3$, $K_{\rm d_1}=0.40;~K_{\rm p_2}=11.34$, $K_{\rm i_2}=54.11$, $K_{\rm d_2}=0.18$. The contour error feedback gains of the CCC are chosen as $K_{\rm x_1}=K_{\rm x_2}=1$; the desired error dynamics matrix of the TCF is chosen as $A_{\rm Td}= [0~0~1~0;~0~0~0~1;~-1{\rm e}^4~0~-1{\rm e}^2~0;~0~-1{\rm e}^4~0~-1{\rm e}^2]$; the TV-IMCC controller is a third-order one ($\xi_1\in\R^2$ and $\xi_2\in\R$), whose time-varying parameters are calculated in real-time by equation~(\ref{se}) and~(\ref{stab}), with stabilizer gains $K=[-1.50{\rm e}^3~-9.47{\rm e}^2]$ solved by the LMI~(\ref{lmi}) and $F=1{\rm e}^{-4}$ in~(\ref{ob}). The contouring errors of different methods are unified within the GCCC criterion~\cite{cheng2009} for a fair comparison. The results are shown in the following subsections. %The experiments of TCF methods are omitted because of the complicated computation in solving coordinate transformation matrices and low-efficient parameter tuning. %Meanwhile, it is noticed that in the TV-IMCC structure, if the two-axis simulations and experiments are successfully conducted, the results can be easily generalized to multi-axis conditions. Thus, simulations and experiments of two-dimensional contours are mainly concerned.
\subsection{Simulation study}
\begin{figure}[!t]
  \centering
  \includegraphics[width=0.5\hsize]{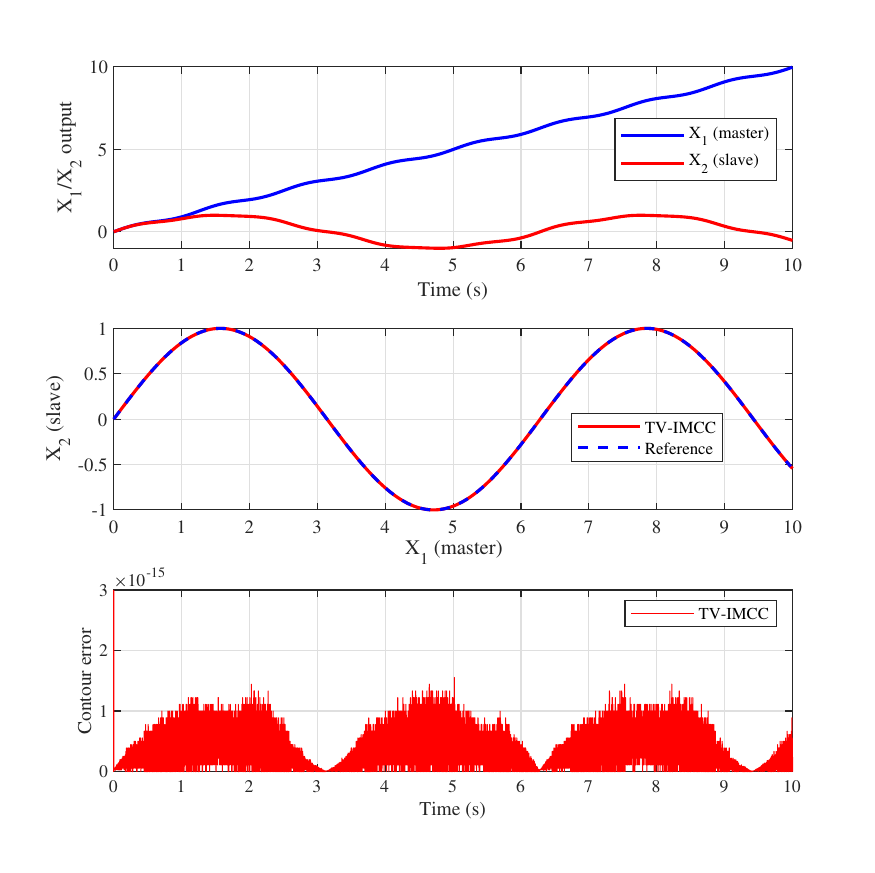}\\
  \caption{Simulation results of a sinusoidal contour tracking by the proposed TV-IMCC.}\label{simuidealsine}
\end{figure}
\begin{figure}[!t]
  \centering
  \includegraphics[width=0.5\hsize]{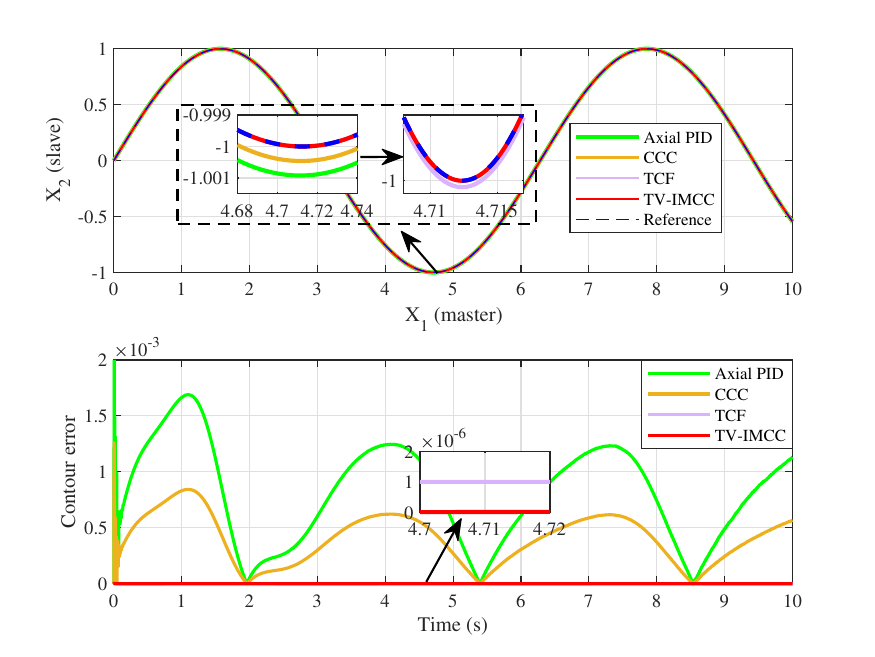}\\
  \caption{Simulation results of tracking a sinusoidal contour by axial PID, CCC, TCF, and TV-IMCC.}\label{simusine}
\end{figure}
\begin{figure}[!t]
  \centering
  \includegraphics[width=0.5\hsize]{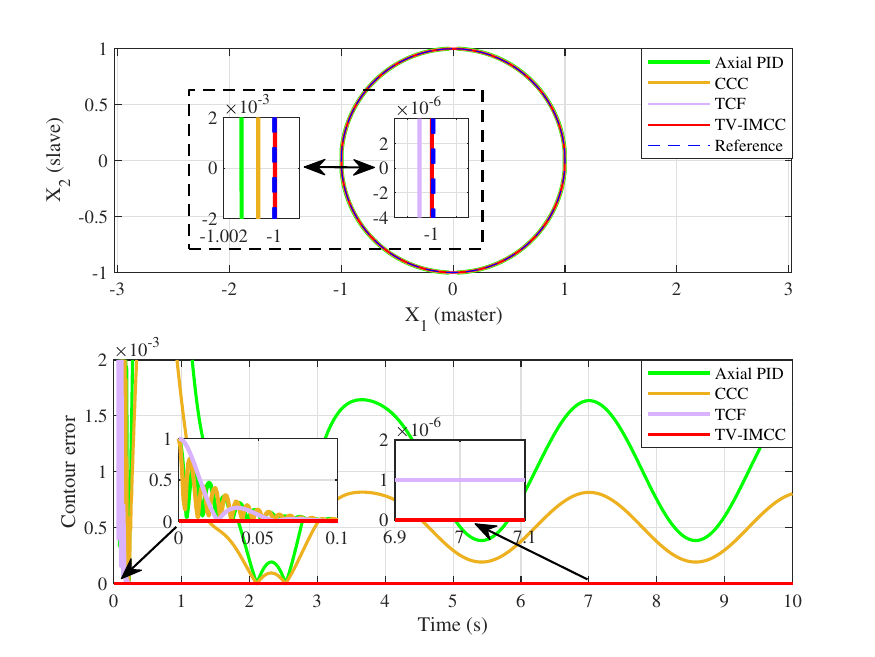}\\
  \caption{Simulation results of tracking a circular contour by axial PID, CCC, TCF, and TV-IMCC.}\label{simucircle}
\end{figure}
The simulations are conducted for the motion stage dynamics with zero initial condition in MATLAB$^\text{TM}$ Simulink. We begin with a sinusoidal contour satisfying Assumption~\ref{a1}, and the reference is $r_2=\sin y_1$, and the position of the master axis $X_1$ is in the form $y_1=t+0.1\sin5t$. To better demonstrate the asymptotical tracking performance of the TV-IMCC, here we introduce a signal $d_{\rm a}=0.1\sin5t$ to mimic an obvious tracking error of master axis. With the structure of $y_1$, the exosystem is obtained via Proposition~\ref{prop:exo} in advance of the simulation. The tracking results of the slave axis $X_2$ and the contour error are shown in Fig.~\ref{simuidealsine}, where the order of the contour error is of $10^{-16}$ (around the precision of floating numbers in MATLAB). The results indicate that the reference is  asymptotically tracked by the TV-IMCC in the presence of a large yet bounded tracking error of the master axis.%, as long as the tracking error of the master axis is bounded. %Indeed, the actual time-varying generating dynamics is reconstructed by the position of the master axis in real-time by equation~(\ref{exomodel}) with a sampling-period-order delay. So the tracking performance in actual situations may deteriorate compared with the ideal situation.

%Taking the above factor into account,
To compare the contouring capability of the proposed TV-IMCC with the axial PID,~CCC and TCF, the following comparison simulations are conducted for three kinds of contours: a sinusoidal contour with $r_1= t$ and $r_2=\sin t$, a circular contour with $r_1=\cos t$ and $r_2=\sin t$, and a heart contour with $r_1=\cos t$ and $r_2=\sin t+(\cos t)^{2/3}$. The exosystems are obtained via Proposition~\ref{prop:exo} in realtime simulations. The results are shown in Fig.~\ref{simusine}, \ref{simucircle} and \ref{simuheart}, respectively. It is seen from the results that the TV-IMCC significantly outperforms other methods with respect to the contour error. Furthermore, given the tracking results of the heart contour shown in Fig.~\ref{simuheart}, we find that the TV-IMCC is able to optimize the tracking performance around the portions with large curvature of the contour, thanks to the fast transient behavior enabled by the internal model principle-based controller.
\begin{figure}[!t]
  \centering
  \includegraphics[width=0.5\hsize]{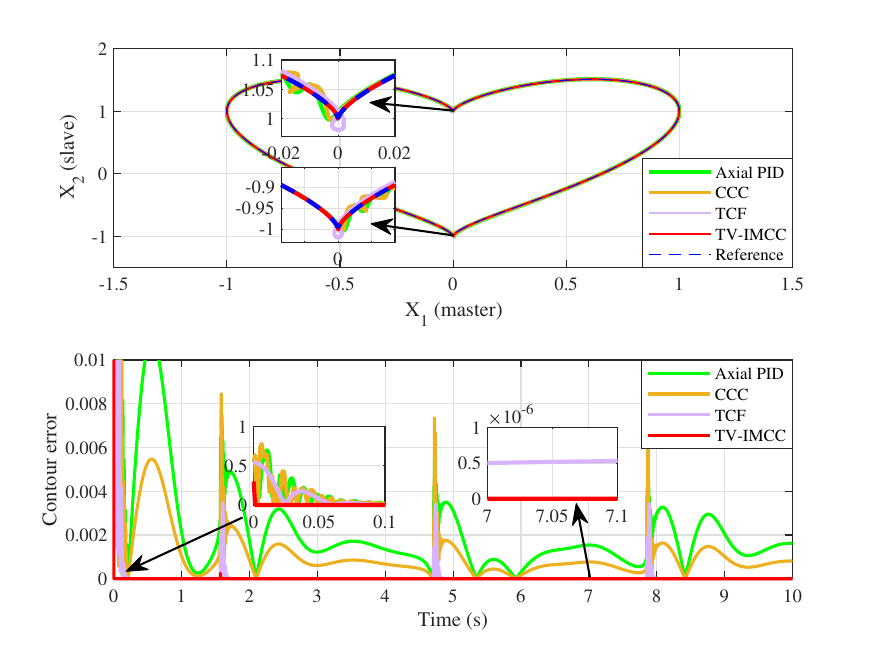}\\
  \caption{Simulation results of tracking a heart contour by axial PID, CCC, TCF, and TV-IMCC.}\label{simuheart}
\end{figure}

To summarize the results of the above comparison simulations, Table~\ref{tabs} containing the root mean square (RMS) and maximum of the simulation contour errors is provided.
\begin{table}[!t]
\tabcolsep=1.5pt
\renewcommand\arraystretch{1.2}
  \centering
     \makeatletter\def\@captype{table}\makeatother\caption{Simulation contouring performance indices: Error RMS and maximum.}
  \begin{tabular}{c|cc|cc|cc|cc}
    \toprule
    Controller & \multicolumn{2}{c|}{PID} & \multicolumn{2}{c|}{CCC} & \multicolumn{2}{c|}{TCF} & \multicolumn{2}{c}{TV-IMCC} \\
    Error index & $\varepsilon_{\rm rms}$ & $\varepsilon_{\rm max}$ & $\varepsilon_{\rm rms}$ & $\varepsilon_{\rm max}$ & $\varepsilon_{\rm rms}$ & $\varepsilon_{\rm max}$ & $\varepsilon_{\rm rms}$ & $\varepsilon_{\rm max}$ \\
    \midrule
    \tabincell{c}{Sinusoidal} & $\scriptstyle 8.7{\text e}^{-4}$ & $\scriptstyle 1.7{\text e}^{-3}$ & $\scriptstyle 4.3{\text e}^{-4}$ & $\scriptstyle 8.4{\text e}^{-4}$ & $\scriptstyle 6.5{\text e}^{-7}$ & $\scriptstyle 1.0{\text e}^{-6}$ & $\scriptstyle 2.6{\text e}^{-11}$ & $\scriptstyle 5.1{\text e}^{-9}$\\
    \tabincell{c}{Circular} & $\scriptstyle 1.1{\text e}^{-3}$ & $\scriptstyle 1.6{\text e}^{-3}$ & $\scriptstyle 5.7{\text e}^{-4}$ & $\scriptstyle 8.2{\text e}^{-4}$ & $\scriptstyle 6.6{\text e}^{-7}$ & $\scriptstyle 1.0{\text e}^{-6}$ & $\scriptstyle 2.6{\text e}^{-11}$ & $\scriptstyle 5.0{\text e}^{-9}$\\
    \tabincell{c}{Heart} & $\scriptstyle 1.6{\text e}^{-3}$ & $\scriptstyle 5.8{\text e}^{-3}$ & $\scriptstyle 9.0{\text e}^{-4}$ & $\scriptstyle 7.4{\text e}^{-3}$ & $\scriptstyle 2.8{\text e}^{-4}$ & $\scriptstyle 3.5{\text e}^{-3}$ & $\scriptstyle 2.7{\text e}^{-7}$ & $\scriptstyle 1.5{\text e}^{-5}$\\
    \bottomrule
  \end{tabular}
   \label{tabs}
\end{table}

Finally, to illustrate the multi-axis contour tracking performance of the TV-IMCC, we design a contour consisting of four axial sinusoidal signals, that is $r_1 = \cos t,\,r_2 = \sin t,\,r_3 = 0.1\cos10t$, and $r_4=0.1\sin10t$. Specifically, $r_1$ and $r_2$ are tracked by the axes of the motion stage ($X_1$, $X_2$); while $r_3$ and $r_4$ are tracked by the galvoscanner ($X_3$, $X_4$). The relative output position $y_1-y_3$ composes the X axis coordinate, and $y_2-y_4$ composes the Y axis coordinate of the contour. In this case, $X_1$ is chosen as the master axis, while $X_2, X_3$ and $X_4$ are the slave axes. The parameters of the scanner models $(G_3,H_3,C_3)$ and $(G_4,H_4,C_4)$ are identified as follows
\begin{equation}
\begin{split}
&G_3=[1.7387\quad1;\quad-0.7529\quad 0],\\
&H_3=[0.0270\quad0.0246]^\top,\quad C_3=[1\quad0];\\
&G_4=[1.7261\quad1;\quad-0.7405\quad 0],\\
&H_4=[0.0253\quad0.0228]^\top,\quad C_4=[1\quad0].
\end{split}
\end{equation}
and the focal distance is set as $100$. The stabilizer gains in~(\ref{ob}) for the scanner are given as $K=[-41.15~-26.41]$ and $F=1{\rm e}^{-4}$. We apply the TV-IMCC to the three slave axes, and the tracking results in Fig.~\ref{simumulti} show that the TV-IMCC is able to achieve asymptotic contour tracking for 4-axis systems. Note that the CCC and TCF methods are not applicable in this case, because the decoupling is not available for different axes in the same direction. Therefore, the proposed TV-IMCC is more feasible and effective for multi-axis contouring compared to the CCC and TCF.
\begin{figure}[!t]
  \centering
  \includegraphics[width=0.5\hsize]{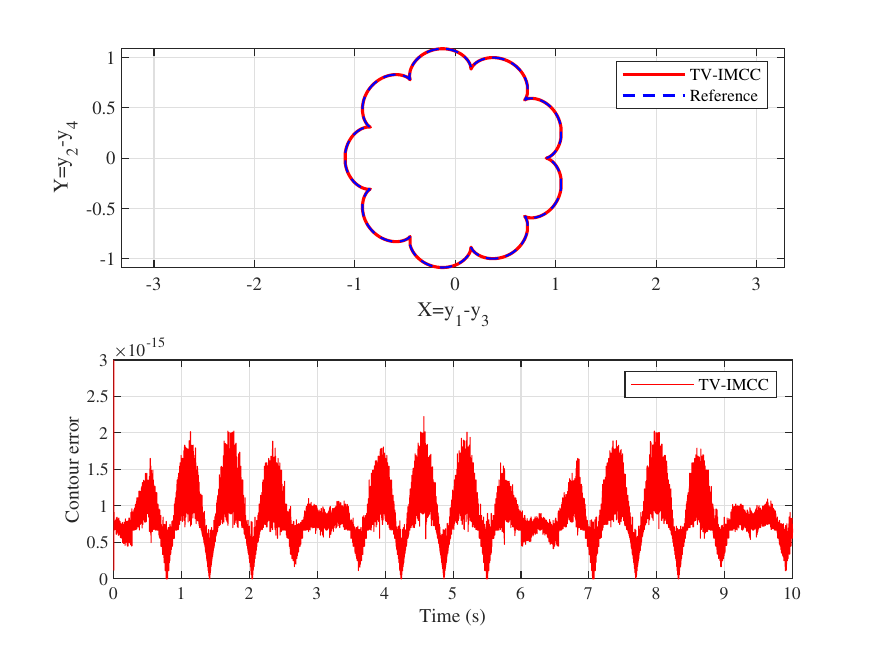}\\
  \caption{Simulation results of tracking a 4-axis contour by TV-IMCC.}\label{simumulti}
\end{figure}
\subsection{Experimental study}
To further demonstrate the feasibility of the proposed TV-IMCC, experiments are also implemented on the system shown in Fig.~\ref{stage}. The detailed parameters of the stage are listed in Table.~\ref{stagep}. A dSPACE$^\text{TM}$ 1103 rapid prototyping system and MATLAB Simulink are used for controller implementation and real-time control executions with the sampling rate of $1$kHz. Note that CCC is not immediately enabled in the experiments for transient stability.

Firstly, a sinusoidal contour %$x_2=\sin 5x_1$
constructed by $r_1=0.1\pi t$ (mm) and $r_2=30\sin0.5\pi t$ (mm) is tracked. From the tracking results in Fig.~\ref{expsine}, it is seen that the PID, CCC and TCF cause error deterioration at certain time intervals, because the axial external disturbances (e.g., the friction) are coupled by the contour error calculation and fedback to each axis. In contrast, the proposed TV-IMCC only controls the slave axis ($X_2$) based on the master axis ($X_1$) position, in other words, the axial external disturbances are not coupled. Hence, there is no such deterioration of the contour error in the tracking results of the TV-IMCC. Furthermore, the TV-IMCC is able to achieve better performance of unknown-structured disturbance rejection in the master axis, according to the simulation results shown in Fig.~\ref{simuidealsine}. For the above reasons, it is seen from Fig.~\ref{expsine} that the contour error of the TV-IMCC is much smaller and more stationary compared to the existing methods.
\begin{figure}[!t]
  \centering
  \includegraphics[width=0.5\hsize]{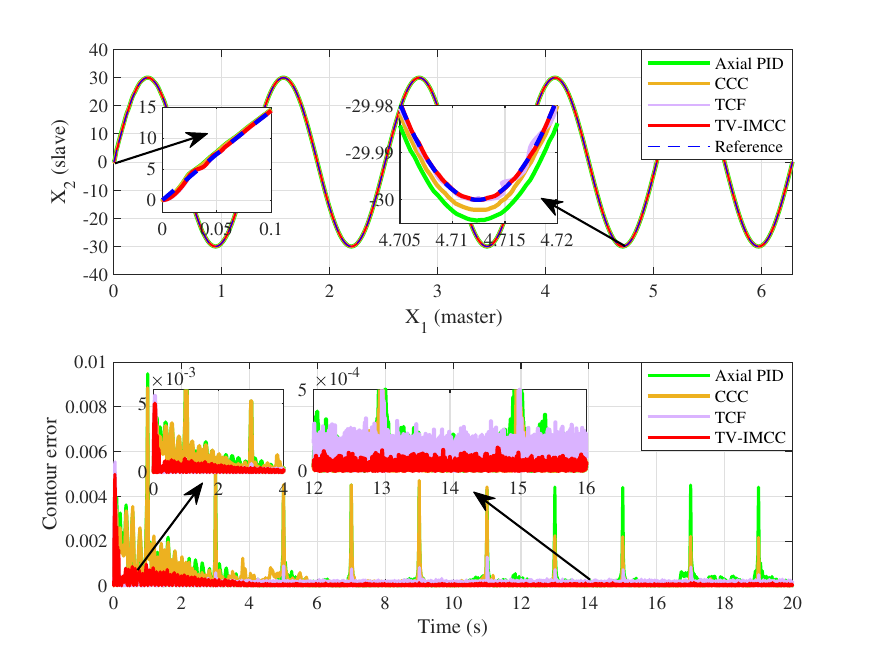}\\
  \caption{Experimental results of tracking a sinusoidal contour by axial PID, CCC, TCF and TV-IMCC.}\label{expsine}
\end{figure}
\begin{table}[!t]
  \centering
     \makeatletter\def\@captype{table}\makeatother\caption{Technical parameters of the XY servo stage.}
  \begin{tabular}{lccc}
    \toprule
    \tabincell{c}{\textbf{Actuators}} & Akribis AUM2-S-S4 \\
    \tabincell{c}{\textbf{Servo drives}} & Copley Accelnet ACJ-090-12 \\
    \tabincell{c}{\textbf{Cross roller guides}} & NB SV4360-35Z \\
    \tabincell{c}{\textbf{Max. stroke}} & $200$mm$\times200$mm \\
    \tabincell{c}{\textbf{Max. velocity}}   & $300$mm/sec for each axis\\
    \tabincell{c}{\textbf{Max. acceleration}}   & $0.5$g for each axis\\
    \tabincell{c}{\textbf{Optical encoders}}   & Renishaw T1011-15A\\
    \tabincell{c}{\textbf{Encoder resolution}}   & $10$nm\\
    \bottomrule
  \end{tabular}
   \label{stagep}
\end{table}
\begin{figure}[!t]
  \centering
  \includegraphics[width=0.5\hsize]{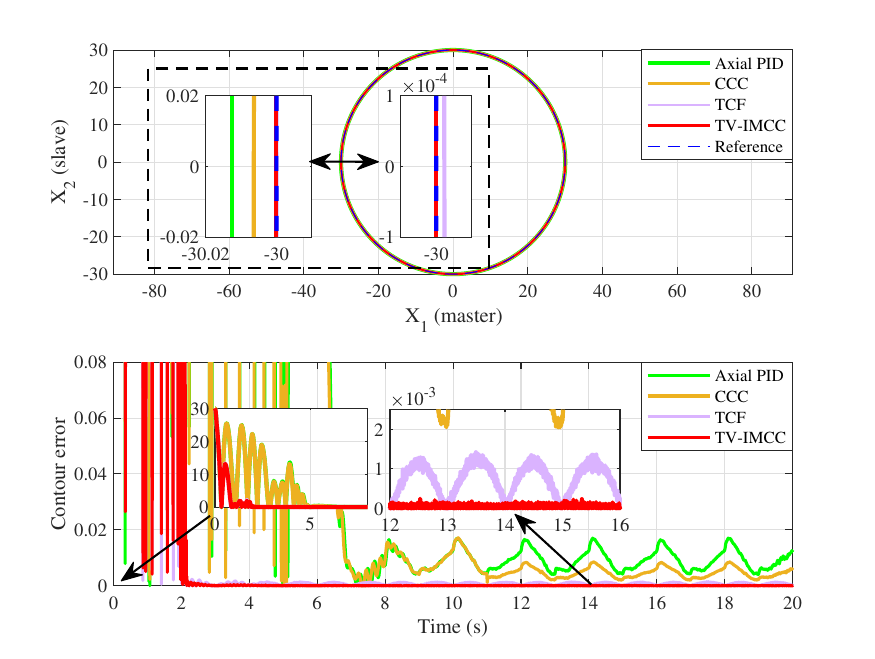}\\
  \caption{Experimental results of tracking a circular contour by axial PID, CCC, TCF and TV-IMCC.}\label{expcircle}
\end{figure}
\begin{figure}[!t]
  \centering
  \includegraphics[width=0.5\hsize]{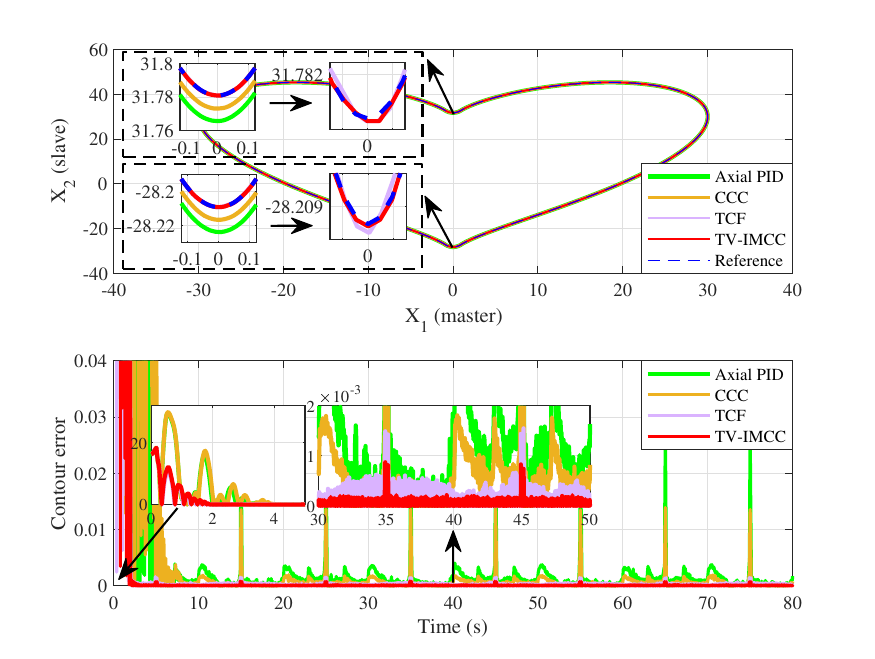}\\
  \caption{Experimental results of tracking a heart contour by axial PID, CCC, TCF and TV-IMCC.}\label{expheart}
\end{figure}

Next, a circular contour with $r_1=30\cos0.5\pi t$ (mm) and $r_2=30\sin0.5\pi t$ (mm) is tracked. From the results shown in Fig.~\ref{expcircle}, it is seen that the contour error of the TV-IMCC is also smaller and more uniform compared to those of the existing methods.

Furthermore, an irregular heart contour with $r_1=30\cos0.1\pi t$ (mm) and $r_2=30(\sin0.1\pi t+(\cos0.1\pi t)^{2/3})$ (mm) is tracked (the contour is smoothed to avoid input saturation). It is seen from Fig.~\ref{expheart} that the proposed TV-IMCC outperforms the existing methods especially around the portions with large curvature of the contour.
\begin{table}[!t]
\tabcolsep=3pt
\renewcommand\arraystretch{1.2}
  \centering
     \makeatletter\def\@captype{table}\makeatother\caption{Experimental contouring performance indices: Error RMS $\varepsilon_{\rm rms}$ and maximum $\varepsilon_{\rm max}$ ($\upmu \rm m$).}
  \begin{tabular}{c|cc|cc|cc|cc}
    \toprule
    Controller & \multicolumn{2}{c|}{PID} & \multicolumn{2}{c|}{CCC} & \multicolumn{2}{c|}{TCF} & \multicolumn{2}{c}{TV-IMCC} \\
    Error index & $\varepsilon_{\rm rms}$ & $\varepsilon_{\rm max}$ & $\varepsilon_{\rm rms}$ & $\varepsilon_{\rm max}$ & $\varepsilon_{\rm rms}$ & $\varepsilon_{\rm max}$ & $\varepsilon_{\rm rms}$ & $\varepsilon_{\rm max}$ \\
    \midrule
    \tabincell{c}{Sinusoidal} & $0.47$ & $4.62$ & $0.23$ & $2.36$ & $0.15$ & $1.44$ & $0.04$ & $0.23$\\
    \tabincell{c}{Circular} & $9.98$ & $17.40$ & $4.99$ & $8.76$ & $0.84$ & $1.54$ & $0.06$ & $0.26$\\
    \tabincell{c}{Heart} & $3.12$ & $28.06$ & $1.53$ & $13.94$ & $0.23$ & $1.59$ & $0.06$ & $0.91$\\
    \bottomrule
  \end{tabular}
   \label{tab1}
\end{table}

The above comparison results of the contouring experiments are consistent with those of the simulations. Moreover, to summarize the experimental contouring indices, the RMS and maximum of the above experimental contour errors are listed in Table~\ref{tab1}. It is obtained that contour error of the TV-IMCC is about two orders of magnitude smaller than that of the PID and CCC, and about one order of magnitude smaller than that of the TCF.

Finally, to demonstrate the multi-axis contouring capability of the TV-IMCC, a 4-axis contouring experiment is conducted on the galvoscanner-stage synchronized motion system shown in Fig.~\ref{stage}. The references $r_1 = 30\cos 0.2\pi t$ (mm) and $r_2 = 30\sin0.2\pi t$ (mm) are for the stage; $r_3 = 3\cos2\pi t$ (mm) and $r_4=3\sin2\pi t$ (mm) are for the galvoscanner. The relative position between the stage and the galvoscanner ($X_{\rm s}-X_{\rm g},\,Y_{\rm s}-Y_{\rm g}$) composes the desired contour. The tracking results are shown in Fig~\ref{exp4axis}, where the contour error (RMS) is $1.60\upmu \rm m$ with equivalent $100$nm resolution of the scanner. Further, the sharp-pointed portions of the contour can be well tracked, thanks to the fast response of the scanner and the synchronized motion manner enabled by the proposed TV-IMCC.
\begin{figure}[!t]
  \centering
  \includegraphics[width=0.5\hsize]{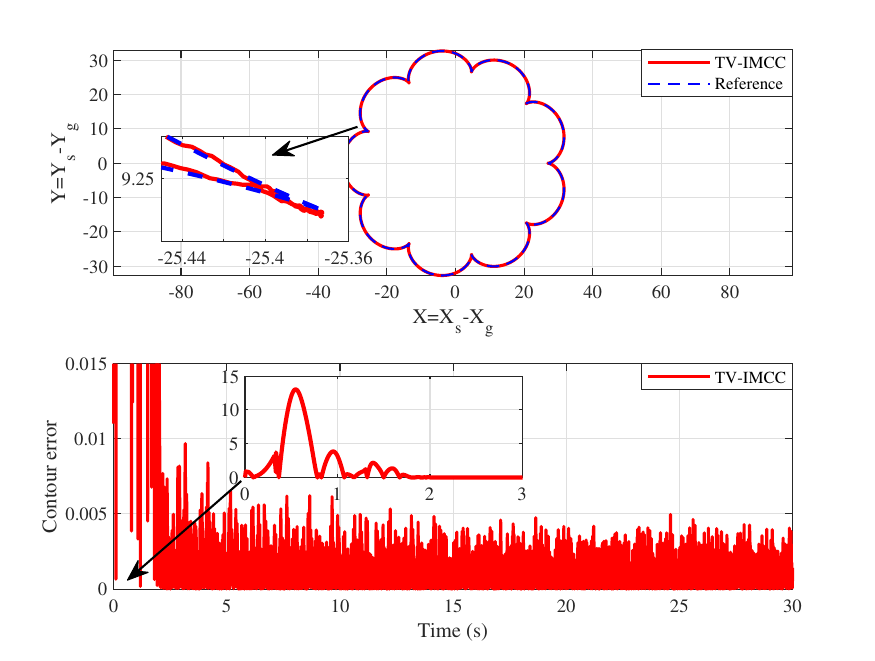}\\
  \caption{Experimental results of tracking a 4-axis contour by the TV-IMCC.}\label{exp4axis}
\end{figure}
\section{Conclusion and future work}\label{conclusion}
In this paper, we have proposed a TV-IMCC methodology to deal with high-precision multi-axis contour tracking.
%Firstly, we have proposed a general signal conversion algorithm which transforms a rotational variable to a monotonic one.
By introducing a general signal conversion algorithm, we have extended a position domain framework such that the multi-axis contour tracking problem is successfully decoupled to individual axial time-varying tracking problem, which is handled by means of a time-varying internal model-based controller.
Furthermore, we have shown the asymptotic tracking performance of the closed-loop system with the proposed TV-IMCC. Numerous simulation and experimental results show that the proposed TV-IMCC outperforms the existing methods for various kinds of contours.

We currently attempt to normalize the augmented system dynamics by utilizing an extended state observer-based design, so that the model uncertainties and external disturbances with unknown structure can be rejected as well.

\section{Acknowledgment}
The authors would like to acknowledge the financial support from the National Natural Science Foundation of China under Grant (52275564, 51875313).

\section{Appendix}
\subsection{Proof of Lemma~\ref{l1}}\label{apx1}
\begin{proof}
Rewrite $\alpha$ as
\begin{equation}\label{l1p1}
\alpha=s\pi+p\,,
\end{equation}
where $p\in [0,\pi]$. Thus
\begin{equation}\label{l1p2}
\begin{array}{rcl}
\cos\alpha&=&\cos(p+s\pi)\\[1mm]
&=&\begin{cases}
\cos p\,,& s~\text{is {\rm even}}\\
\cos (\pi-p)\,,& s~\text{is {\rm odd}}\\
\end{cases}\\[1mm]
&=&\cos(\dfrac{\pi}{2}(1-(-1)^s)+(-1)^sp)\,.
\end{array}
\end{equation}
Substitute~(\ref{l1p2}) into~(\ref{ls1}), yielding
\begin{equation}\label{l1p3}
\begin{array}{rcl}
\alpha&=&s\pi+\dfrac{\pi}{2}(1-(-1)^s)+(-1)^s\arccos(\cos\alpha)\\[3mm]
&=&s\pi+\dfrac{\pi}{2}(1-(-1)^s)\\[3mm]
&&+(-1)^s(\dfrac{\pi}{2}(1-(-1)^s)+(-1)^sp)\\[3mm]
&=&s\pi+p\,,
\end{array}
\end{equation}
which completes the proof.\hfill $\square$
\end{proof}
\subsection{Proof of Proposition~\ref{prop:solutioneqn5}}\label{apx:solutioneqn5}
\begin{proof}
Divide $\theta$ into intervals on $\R$ as
\begin{equation}\label{t1p1}
\theta\in\bigcup_{s\in \N}[s\pi,(s+1)\pi]\,.
\end{equation}
Then, it remains to show the following proposition
\begin{equation}\label{t1p2}
\theta_{\rm e}(k)=\theta(k)\,,~\theta(k)\in[s\pi,(s+1)\pi]\,,
\end{equation}
for any $s\in \N$. We prove it by an induction on $s$. According to Lemma~\ref{l1}, when $s=0$, the proposition is clearly true. Assume that for a particular $n\in\N$, the proposition is true when the single case $s=n$ holds. Hence,
\begin{equation}\label{t1p22}
\theta_{\rm e}(k)=n\pi+\frac{\pi}{2}(1-(-1)^n)+(-1)^n\arccos(\cos\theta(k))=\theta(k)\,,
\end{equation}
and $\theta(k)\in[n\pi,(n+1)\pi]$. Equation~(\ref{t1p22}) can be rewritten as
\begin{equation}\label{t1p3}
\theta_{\rm e}(t)=n\pi+\frac{\pi}{2}(1-(-1)^n)+(-1)^n\arccos(\cos \theta(t))\,.
\end{equation}
The derivative of~(\ref{t1p3}) can be obtained as
\begin{equation}\label{t1p4}
\dot \theta_{\rm e}(t)=(-1)^n\sign{\sin\theta}\dot \theta(t)\,.
\end{equation}
Since $\theta(t)$ is strictly monotonic, it can be obtained that
\begin{equation}\label{t1p44}
(-1)^n\sign{\sin\theta}>0\,,
\end{equation}
when $\theta\in[n\pi,(n+1)\pi]$. Thus, it is easy to know that
\begin{equation}\label{t1p5}
(-1)^n\sign{\sin\theta}<0\,,
\end{equation}
when $\theta\in[(n+1)\pi,(n+2)\pi]$. So, if $\theta_{\rm e}(k)$ is obtained by~(\ref{t1p22}) for $\theta\in[(n+1)\pi,(n+2)\pi]$, $\theta_{\rm e}(k)<\theta_{\rm e}(k-1)$ occurs when the interval of $\theta$ varies. According to Algorithm~\ref{alg1}, index $s$ will be increased by $1$ at this $k$, yielding
\begin{equation}\label{t1p6}
\begin{array}{rcl}
\theta_{\rm e}(k)&=&(n+1)\pi+\dfrac{\pi}{2}(1-(-1)^{n+1})\\[2mm]
&&+(-1)^{n+1}\arccos(\cos\theta(k))\\[2mm]
&=&\theta(k)\,,
\end{array}
\end{equation}
for $\theta(k)\in[(n+1)\pi,(n+2)\pi]$. Therefore the proposition holds for any $s\in \N$, and this completes the proof.\hfill $\square$
\end{proof}
\subsection{Proof of Proposition~\ref{prop:exo}}\label{apx:exo}
\begin{proof}
With equation~(\ref{l}), the variables $y_1$ and $r_2$ can be rewritten in the following form
\begin{equation}\label{x1r2}
\begin{array}{rcl}
y_1&=&l(y_1)\cdot\dfrac{y_1}{\sqrt{y_1^2+f^2(y_1)}}\\[5mm]
r_2&=&l(y_1)\cdot\dfrac{f(y_1)}{\sqrt{y_1^2+f^2(y_1)}}\,.
\end{array}
\end{equation}
According to equation~(\ref{eta}), equation~(\ref{x1r2}) can be rewritten as
\begin{equation}\label{x1r2eta}
\begin{array}{rcl}
y_1&=&l(y_1)\sin\eta(y_1)\\[1mm]
r_2&=&l(y_1)\cos\eta(y_1)\,.
\end{array}
\end{equation}
It is clear that
\begin{equation}\label{x1r2deriv}
\begin{array}{rcl}
r'_2&=&l'(y_1)\cos\eta(y_1)-\eta'(y_1)l(y_1)\sin\eta(y_1)\\[1mm]
y'_1&=&l'(y_1)\sin\eta(y_1)+\eta'(y_1)l(y_1)\cos\eta(y_1)\,.
\end{array}
\end{equation}
Substitute~(\ref{x1r2eta}) into~(\ref{x1r2deriv}), yielding
\begin{equation}\label{x1r2state}
\begin{array}{rcl}
r'_2&=&\dfrac{l'(y_1)}{l(y_1)}r_2-\eta'(y_1)y_1\\[5mm]
y'_1&=&\eta'(y_1)r_2+\dfrac{l'(y_1)}{l(y_1)}y_1\,.
\end{array}
\end{equation}
Let
\[
w(y_1)=\begin{bmatrix}r_2\\y_1\end{bmatrix}\,,
\]
then equation~(\ref{x1r2state}) can be rewritten as
\begin{equation}\label{x1r2ss}
w'(y_1)=\begin{bmatrix} \dfrac{l'(y_1)}{l(y_1)} & -\eta'(y_1)\\\eta'(y_1) &\dfrac{l'(y_1)}{l(y_1)}\end{bmatrix}w(y_1).
\end{equation}
It is obvious that $S(y_1)$ and $Q(y_1)$ are of the form~(\ref{exomodel}), which completes the proof.~\hfill $\square$
\end{proof}
\subsection{Proof of Lemma~\ref{polytype}}\label{apx2}
\begin{proof}
Based on the Lyapunov inequality in the time-varying system, it is known that the feedback gain $K(k)$ can be designed by finding a symmetric positive-definite matrix $P(k)$ satisfying
\begin{equation}\label{lyatv}
A_c^{\rm T}(k)P(k+1)A_c(k)-P(k)<0\,,
\end{equation}
where $A_c(k)=A(k)+B(k)K(k)$. To make the presentation concise, we drop index $k$, while keeping $k+1$ in the corresponding matrices. From inequality~(\ref{lyatv}), note that
\begin{equation}\label{dv3}
\begin{split}
&P-A_{\rm c}^{\rm T}P(k+1)A_{\rm c}\\[2mm]
&=P-\left[A_{\rm c}^{\rm T}P^{\rm T}(k+1)\right]P^{-1}(k+1)\left[P(k+1)A_{\rm c}\right]>0\,.
\end{split}
\end{equation}
Applying the Schur complement to inequality~(\ref{dv3}) yields
\begin{equation}\label{lmi2}
\begin{bmatrix} P& A_c^{\rm T}P(k+1) \\P(k+1)A_{\rm c} & P(k+1)\end{bmatrix}>0\,.
\end{equation}
Therefore, the rest of the proof is to solve inequality~(\ref{lmi2}) by using Proposition~\ref{a3}, which can be referred to Theorems~3 and~4 in~\cite{daafouz2001}.\hfill $\square$
\end{proof}

\bibliographystyle{unsrt}

\bibliography{references}

\end{document}